\documentclass[journal,onecolumn,draftclsnofoot]{IEEEtran}
\usepackage[pdftex]{graphicx}
 \pdfoutput=1
\usepackage{enumerate}

\usepackage{caption}
\usepackage{footnote}
\usepackage{subcaption}
\usepackage{amsthm}
\usepackage{amssymb}
\usepackage{dsfont}
\usepackage{setspace}
\usepackage[linesnumbered,ruled,vlined]{algorithm2e}
\usepackage{comment}
\usepackage{color}
\usepackage{float}
\usepackage{gensymb}
\usepackage{cuted}
\usepackage{soul}
\usepackage{cite}
\theoremstyle{plain}

\newtheorem{theorem}{Theorem}
\newtheorem{lemma}{Lemma}
\newtheorem{fact}{Fact}
\newtheorem{definition}{Definition}

\theoremstyle{definition}

\theoremstyle{remark}
\newtheorem{remark}{Remark}

\usepackage{mathtools}
\newcommand{\expect}{\operatorname{\mathbb{E}}\expectarg}
\DeclarePairedDelimiterX{\expectarg}[1]{[}{]}{%
  \ifnum\currentgrouptype=16 \else\begingroup\fi
  \activatebar#1
  \ifnum\currentgrouptype=16 \else\endgroup\fi
}
\newcommand{\prob}{\operatorname{\mathbb{P}}\probarg}
\DeclarePairedDelimiterX{\probarg}[1]{(}{)}{%
  \ifnum\currentgrouptype=16 \else\begingroup\fi
  \activatebar#1
  \ifnum\currentgrouptype=16 \else\endgroup\fi
}
\newcommand{\innermid}{\nonscript\;\delimsize\vert\nonscript\;}
\newcommand{\activatebar}{%
  \begingroup\lccode`\~=`\|
  \lowercase{\endgroup\let~}\innermid 
  \mathcode`|=\string"8000
}

\def\*#1{\mathbf{#1}}
\def\&#1{\mathcal{#1}}
\def\.#1{\boldsymbol#1}
\def\^#1{\hat{#1}}
\def\[#1{\left #1}
\def\]#1{\right #1}

\DeclareMathOperator*{\argmin}{arg\,min\ }
\DeclareMathOperator*{\argmax}{arg\,max\ }

\DeclareMathAlphabet\mathbfcal{OMS}{cmsy}{b}{n}

\author{\IEEEauthorblockN{Sung-En Chiu and Tara Javidi}
\thanks{This paper was presented in part at Information Theory Workshop 2016.}
\vspace{2mm}

\IEEEauthorblockA{Department of Electrical and Computer Engineering \\
University of California, San Diego \\
Email: \{suchiu,tara\}@ucsd.edu}
}

\newcommand{\rv}[1]{{\color{black}{{#1}}}}

\begin{document}

\onehalfspacing

\title{\rv{Low Complexity Sequential Search with Size-Dependent Measurement Noise}}


\maketitle

\begin{abstract}
This paper considers a target localization problem where at any given time an agent can choose a region to query for the presence of the target in that region. The measurement noise is assumed to be increasing with the size of the query region the agent chooses. Motivated by practical applications such as initial beam alignment in array processing, heavy hitter detection in networking, and visual search in robotics, we consider practically important complexity constraints/metrics:  \textit{time complexity}, \textit{computational and memory complexity}, \rv{and the complexity of possible query sets in terms of geometry and cardinality.} 

Two novel search strategy, $dyaPM$ and $hiePM$, are proposed. \rv{Pertinent to the practicality of out solutions}, $dyaPM$ and $hiePM$ are of a connected query geometry (i.e. query set is always a connected set) implemented with low computational and memory complexity. Additionally, $hiePM$ has a hierarchical structure and, hence, a further reduction in the cardinality of possible query sets, making $hiePM$ practically suitable for applications such as beamforming in array processing where \rv{memory limitations favors a smaller codebook size}.

Through a unified analysis with Extrinsic Jensen Shannon (EJS) Divergence, $dyaPM$ is shown to be asymptotically optimal in search time complexity (asymptotic in both resolution (rate) and error (reliability)). On the other hand, $hiePM$ is shown to be near-optimal in rate. In addition, both $hiePM$ and $dyaPM$ are shown to outperform prior work in the non-asymptotic regime.  

\end{abstract}
\section{Introduction}

\setcounter{footnote}{1}

We consider a target search problem where at any given time, an agent can choose a query set inspected for the presence of the target. More precisely, upon querying a set, the agent receives a noisy measurement indicating the presence of the target in the set. The agent conducts multiple queries where each query set can, in general, be chosen adaptively and strategically based on previous (noisy) measurements. The main focus of this paper is to design and analyze \rv{low complexity search strategies} under a realistic model. Motivated by many practical applications, such as spectrum sensing \cite{Ronquillo2017} in cognitive radio, Angle-of-Arrival (AoA) estimation in initial beam alignment \cite{Chiu2019}, and heavy hitter detection in networking \cite{Wang2018}, \rv{the measurements are assumed to be subject to an additive noise whose statistics is assumed to be dependent on the \textit{size} of the inspection region}. More precisely, querying a larger region results in a noisier measurement than querying a smaller region. With binary measurements and Bernoulli noise, for instance, this would mean that the false alarm and miss detection of each query is a non-decreasing function of the size of the query set. 

\subsection{\rv{Literature Review}}
The problem of binary noisy search for a single target with measurement-independent noise has been studied extensively \cite{Horstein63,burnashev1974interval,Borgstrom1993,Karp2007,Or2008,Naghshvar2013ISIT,Naghshvar2013ASH,Naghshvar2015,Henderson2013,Li2014,Shayevitz2016} in the literature. Relying on connections with feedback coding, the authors in \cite{Horstein63,burnashev1974interval} proposed a noisy variant of the binary search algorithm, where the query area is designed to partition the posterior at its median. Posterior Matching (PM) strategy proposed in \cite{Shayevitz11} generalizes the noisy binary search algorithm and its analysis to a general DMC case (we refer to the noisy variant of binary search algorithm as medianPM). In particular, \cite{Shayevitz11} established the rate-optimality where the targeting rate is defined as the asymptotic ratio of the logarithm of the search resolution over the number of queries. Relying on a connection to hypothesis testing \cite{Naghshvar2013ISIT} and \cite{Naghshvar2013ASH} proposed schemes that achieve the optimal rate-reliability trade-off. By allowing for a random search time, \cite{Naghshvar2015} characterized the reliability of medianPM, where reliability is defined to be the asymptotic ratio of the logarithm of error probability over the (expected) number of queries. \rv{In \cite{Henderson2013}, the author analyzed the mean square error of the target location and its estimate over time with medianPM over continuous target. In \cite{Li2014}, the author proposed a randomized Posterior Matching in the context of channel coding with feedback and anaylzed the error exponent of the proposed feedback codes. A hierarchical query set was used in a control sensing setting with communication constraint \cite{Simsek2004}, where the authors proposed a bit-by-bit sensing algorithm. 

In practical applications, the measurement noise statistics may depend on the query set. For instance, the measurement error may increase when the query set is close to the target. This \textit{distance-dependent} measurement noise has been considered in the context of active learning \cite{Nowak2011,Yan2015,Yan2016}. In particular, medianPM generalizes well in distance-dependent noise and is shown to be near optimal in \cite{Nowak2011}. In other applications such as beamforming in wireless communication \cite{Alkhateeb2014}, visual search in robotics\cite{Naghshvar2013ISIT}, and heavy hitter detection in Networking \cite{Wang2018}, measurement noise may be \textit{size-dependent}, i.e. querying a larger region results in a noisier measurement than querying a smaller region. In contrast to the case of distance-dependent noise, medianPM performs relative poorly in the case of size-dependent noise. The problem of noisy search with size-dependent measurement noise was first explicitly formulated in~\cite{Naghshvar2013ISIT}, where the author proposed a search strategy, maxEJS, which can be shown to be order optimal (logarithmically) in time complexity. MaxEJS selects the query set by maximizing a symmetrized divergence, known as Extrinsic Jensen Shannon (EJS) divergence (a function of the posterior) exhaustively over all possible query sets. However, the prohibitive computational complexity of the exhaustive maximization of EJS divergence renders maxEJS impractical in many applications. In a setting which subsumes this formulation, via a sequence of hierarchical queries on a decision tree for active hypothesis testing of anomaly/rare event detection, \cite{Cohen2015,Wang2018} proposed a random-walk based search strategy (referred to as IRW). We note that one of the proposed algorithm in this paper, $hiePM$, is strongly motivated by \cite{Cohen2015} and \cite{Wang2018}. The formulation in \cite{Wang2018} can be specialized to our setting where searching the dyadic intervals can be associated with a decision tree and observation statistics can be arranged on levels of the decision tree. However, the IRW algorithm and its analysis in \cite{Wang2018}, when specialize to our setting, fails to achieve an optimal acquisition rate. 

Non-adaptive optimal search strategies under size-dependent measurement noise were studied in \cite{Kaspi2014,Kaspi2015,Kaspi2018,Zhou2019} where \cite{Kaspi2014,Kaspi2018} used random coding based strategy with the input distribution (size of the query set) optimized for size-dependent measurement noise. \cite{Kaspi2015} extends the algorithm to multiple targets. \cite{Zhou2019} further provided finer analysis on the non-asymptotic bounds and extended the random coding based algorithm on general DMC.} The first optimal (in terms of rate-reliability) \textit{adaptive} search strategy was also proposed in \cite{Kaspi2014}, consisting of three phases of random search strategies; in this paper, we use the shorthand 3rand to refer to this algorithm. By allowing the second and third phases of the search to adapt to the outcome of the previous phase(s), the algorithm was shown to significantly outperform all non-adaptive strategies (in terms of both rate and reliability). This is in sharp contrast to the case of noisy search with measurement-independent noise where randomized non-adaptive searches are known to be asymptotically (rate) optimal\footnote{This is nothing but a manifestation of Shannon's original analysis establishing that feedback cannot increase the capacity of DMC}. While 3rand strategy is optimal in asymptotic sense, it suffers from three main shortcomings due to its essential reliance on random code constructions. Firstly, the algorithms computational complexity (decode/detection complexity) is rather prohibitive. Secondly, the query geometry is not constrained where the collection of possible query sets grows exponentially (in resolution).  Thirdly, by construction, 3rand does not fully utilize the profile of the noise statistics, resulting in a rather poor non-asymptotic performance with significant sensitivity to the choice of hyper parameters associated with the three phases of the algorithm.

\rv{
Our problem also can be viewed as a special case of noisy group testing \cite{Hwang1972,Baldassini2013,Chan2014,Hou2017,Scarlett2019,Kaspi2018} with the caveat where the noise depends on the size of the tested group. While our goal is to recover a single target, group testing in general concerns multiple target identification. For noisy non-adaptive group testing, we refer readers to \cite{Baldassini2013,Chan2014} where the asymptotically optimal rate and algorithms are relatively well-understood. In contrast, much less is known about adaptive group testing \cite{Hwang1972,Kaspi2015,Hou2017,Scarlett2019,Cuturi2020}, despite significant adaptivity gain \cite{Kaspi2015,Scarlett2019,Cuturi2020}. The author in \cite{Hwang1972} generalized the binary search to multiple target in noiseless setup, while \cite{Scarlett2019} addressed noisy adaptive group testing with measurement-independent noise where the main focus is the asymptotic performance of the proposed 2-stage and 3-stage algorithm under different scaling assumption of the number of targets vs number of population. Most relevantly, \cite{Kaspi2015} also generalized the 3rand algorithm to multiple targets with size-dependent measurement noise where a similar 3-stage adaptive algorithm is proposed with provable optimal asymptotic performance. The proposed algorithm in \cite{Kaspi2015} and \cite{Scarlett2019} thus reduces to the 3rand algorithm of \cite{Kaspi2014} when considering only a single target. However, the numerical comparison between our proposed algorithms $hiePM$ and $dyaPM$ vs 3rand for single target suggests that, those limited adaptive multi-stage algorithm with random coding in the first phase, falls short in the non-asymptotic performance. The same conclusion of non-asymptotic benefit of a fully adaptive algorithm, is also observed numerically by \cite{Cuturi2020} where the author employed a Monte Carlo simulated posterior with high-complexity optimization over all possible tests. Extending our low-complexity and deterministic algorithms to multiple targets, albeit beyond the scope of this paper, is an open but important area of research.

Finally we note that our work is also connected to the problem of sparse vector recovery with noisy measurement (also known as noisy Compressed Sensing \cite{Ji2008,Haupt2009,Jin2013,Malloy2014,Pal2015,Atia2012}), with the caveat of adaptive binary measurements \cite{Atia2012} and the objective of support recovery \cite{Jin2013,Pal2015} instead of the vector itself. Our work can, specifically, be viewed as an extension of the adaptive scheme of \cite{Haupt2009} and \cite{Malloy2014} with the caveat of boolen (binary) measurement matrix. In particular, we see a strong connection to the adaptive non-binary sensing algorithm of \cite{Malloy2014} with  multi-stage shrinkage of the sensing area.
}

\subsection{\rv{Our Contributions}}
In this paper, we study the problem noisy search under \rv{size-dependent measurement noise} with the following practical complexity constraints: \textit{query time complexity}, \textit{computational and memory complexity}, \textit{query geometry (the shape of the query set)} and \textit{query cardinality (cardinality of the collection of query sets that are allowed to be chosen from)}. In particular, we proposed three novel fully adaptive sequential search strategies: sorted Posterior Matching ($sortPM$), dyadic Posterior Matching ($dyaPM$) and hierarchical Posterior Matching ($hiePM$). We analyze these algorithms by quantifying the step-by-step extrinsic Jensen-Shannon divergence. In particular, $sortPM$ is shown to optimize the asymptotic time complexity; $dyaPM$ also achieves similar time complexity while ensuring that the agent can only query connected sets. Lastly, we design and analyze $hiePM$ by further limiting the query sets to be the sets restricted to those that can be represented as a decision tree, such as the bisection search set. This property of $hiePM$ means that the collection of query sets forms a hierarchical cover of the search area, and, hence, is of small cardinality; we also note that limiting the query sets makes $hiePM$ particularly suitable for applications such as beamforming in array processing that requires pre-construction and storage of the possible query sets \cite{Chiu2019}.

In addition to the desirable time complexity, all three strategies, we show, can be implemented with low computational and memory complexities. In particular, we provide an exponential improvement in computational and memory complexities of the proposed strategies over the search strategies in the literature of measurement-dependent noisy search (maxEJS and 3rand). Furthermore, through a set of practically motivated numerical examples, we show that all the proposed search strategies have superior non-asymptotic performance compared with that of 3rand. Notably, we demonstrate superior performance of $hiePM$ despite the lack of theoretical guarantee regarding its asymptotic optimality. 

\rv{
Finally, to better position our work, we provide the following highlights of comparison with existing works:
\begin{enumerate}[1.]
    \item Our proposed algorithms, derived from Posterior Matching, have superior performance compared with medianPM under size-dependent measurement noise (In contrast, medianPM without any modification is shown to be near-optimal under distance-dependent measurement noise \cite{Nowak2011})
    \item Our analysis of the posterior shrinkage over time under each of the proposed algorithms by the use of the sorted-coarse-binned log-likelihood and the nested log-likelihood closes the gap of rate optimality of the fully adaptive algorithm, $sortPM$ and $dyaPM$, under size-dependent measurement noise
    \begin{itemize}[-]
        \item[-] we note that the EJS analysis of \cite{Naghshvar2015,Naghshvar2013ISIT}, provided originally in the context of hypothesis testing and channel coding, does not extend the optimality of EJS or its variants to the problem of search with size-dependent measurement noise. See Appendix \ref{proof-overview} for more details
    \end{itemize}
    \item Our proposed algorithms are much more practical compared with maxEJS \cite{Naghshvar2015} and 3rand \cite{Kaspi2014,Kaspi2018} in terms of computational and memory complexity as well as non-asymptotic performance
    \item The rate and reliability of $hiePM$, though in general sub-optimal, is superior to other strategies over hierarchical query set \cite{Simsek2004,Cohen2015,Wang2018} 
    \begin{itemize}
        \item[-] in section \ref{sec:asym}, we specialize the finding from \cite{Wang2018} and show that our rate analysis of the proposed $hiePM$ has a significant improvement over the proposed algorithm IRW from \cite{Wang2018}
        \item[-] we establish the improved performance of $hiePM$ over the scheme proposed and analyzed in \cite{Simsek2004} under measurement-independent noise (since we intent to focus on size-dependent measurement noise in this paper, we refer readers to \cite{Chiu2018} for the comparison between $hiePM$ and the algorithm of \cite{Simsek2004})
    \end{itemize}
\end{enumerate}
}

\subsection{Notations} \label{sec:notation}
We use boldface letters to represent vectors. We write $\.\pi^{\downarrow}$ to denote sorted element of the vector $\.\pi$ in descending order, $i.e.$, $\pi^{\downarrow}_i$ represents the $i$th largest element of $\.\pi$. For a set of indices $S$, we write $\pi_S \equiv \sum_{i\in S} \pi_i$. We denote the space of probability mass functions on set $\mathcal{X}$ as $P(x)$. We denote the Kullback-Leibler (KL) divergence between distribution $P$ and $Q$ by $D(P \| Q) = \sum_x P(x) \log \frac{P(x)}{Q(x)}$. The mutual information between random variable $X$ and $Y$ is defined as $I(X;Y) = \sum_{x,y} p(x,y) \log \frac{p(x,y)}{p(x)p(y)}$, where $p(x,y)$ is the joint distribution, and $p(x)$ and $p(y)$ are the marginals of $X$ and $Y$. Let Bern$(p)$ denote the Bernoulli distribution with parameter $p$, and $I(q,p)$ denote the mutual information of the input $X \sim$ Bern$(q)$ and the output $Y$ of a BSC channel with crossover probability $p$. Let $C_1(p) := D( \text{Bern}(p) \| \text{Bern}(1-p))$. Let $\expect{ \cdot }$ denote the expectation. We use $|S_t|$ to represent the counting measure of a discrete set $S_t$ (cardinality of $S_t$). 


\section{Problem Setup}
\label{prelim}

We consider the problem of searching for a point target in a unit interval \rv{with a given resolution $\delta$. The target is uniformly placed over $1\slash\delta$ bins each of size $\delta$ in this unit interval.}  Without loss of generality, we assume that $1\slash\delta$ to be an integer. Let \rv{$\theta\in \{1,2,...,1\slash\delta\}$} be the index of the bin that contains the target. 

We wish to estimate $\theta$ by sequentially choosing (possible random) query sets $S_t \subseteq \{1,2,...,\frac{1}{\delta}\}$, $t=1,2,...,\tau$. The agent obtains a noisy observation $Y_t$ modelled as 
\begin{equation}
    Y_t = \mathds{1}( \theta \in S_t) \oplus Z_t(S_t),
\end{equation}
where $\oplus$ denotes exclusive OR operation, and $Z_t(S_t)$ is a Bernoulli noise random variable whose statistics depends on the query set $S_t$. In particular, we assume that $Z_t(S_t) \sim \text{Bern}(p(\delta|S_t|))$ where $p: (0,1) \to (0,\frac{1}{2})$ is a continuous and non-decreasing function. We assume that the noise is conditionally (conditioned on $S_t$) i.i.d. across time.

A sequential causal strategy selects random query set $S_t$ as a measurable random variable of the past decisions and observations $(S_1^{t-1},Y_1^{t-1})$, and makes a declaration of the estimate $\hat{\theta}$ at a stopping time $\tau$. After $\tau$ queries, the agent declares the target index $\hat{\theta}$. The search is said to have \textit{resolution} $\frac{1}{\delta}$ and \textit{reliability} $\epsilon$ if 
\rv{
\begin{equation}
\prob{  |\hat{\theta}_\tau \neq \theta|} \geq 1-\epsilon.
\end{equation}
}
We say a strategy is \textit{fixed-length} if stopping time $\tau$ is a deterministic time which is selected independently of the observation sequences $Y_1^{t-1}$; otherwise, we say it is \textit{variable-length}. We say a strategy is $\textit{non-adaptive}$ if the selection of $S_t$ is made independently of the realization of the past observations; this is in contrast to a strictly \textit{adaptive policy} where the selection of $S_t$ explicitly depends on the observation sequence $Y_1^{t-1}$ as well as the past decisions of query set $S_1^{t-1}$.

In this work, we consider the case where the decision has zero initial side information about the location of the target, hence, a uniform Bayesian prior $\pi_i(0) := \prob{\theta = i} = \delta$. By this Bayesian framework, the belief vector $\.\pi(t)$, where its $i_{th}$  component at time $t$ is given as 
\begin{equation}
    \pi_i(t) = \prob{ \theta = i  | Y_1^{t-1}, S_1^{t-1} },
\end{equation}
is a sufficient statistics. In other words, any deterministic stationary adaptive strategy can be denoted by a function
\begin{equation}
    \gamma: \Delta_{\delta} \to 2^{ \{1,2,...,  \frac{1}{\delta}\} }
\end{equation}
where $\Delta_{\delta}$ is the probability simplex of dimension $\frac{1}{\delta}$.  

We characterize the performance of search strategies by the following:
\begin{enumerate}[i)]
    \item \underline{\textit{The Query Time Complexity}}: \\
    We are interested in search strategies that can find the target location accurately (with resolution $\frac{1}{\delta}$) and reliability (with confidence $1-\epsilon$) as quickly as possible. We measure the asymptotic time complexity by how the (expected) number of queries, $\tau_{\epsilon, \delta}$, scales with the resolution $\frac{1}{\delta}$ and the reliability $\epsilon$. Additionally, we use the rate-reliability pair to capture the asymptotic query time complexity:
\begin{definition} \label{defS:rate-reliability}
    A family of search strategies $\gamma_{\epsilon,\delta} $ with resolution $\frac{1}{\delta}$, reliability $\epsilon$, and stopping time $\tau_{\epsilon,\delta}$ are said to achieve a maximum rate R and a maximum reliability E respectively if and only if
    \begin{equation}
    R= \lim_{\delta \rightarrow 0} \frac{ \log (\frac{1}{\delta}) }{\expect{\tau_{\epsilon,\delta}}}, \quad E = \lim_{\epsilon \rightarrow 0} \frac{ \log (\frac{1}{\epsilon}) }{\expect{\tau_{\epsilon,\delta}}}.
    \end{equation}
\end{definition}
    \item \underline{\textit{The Computational and Memory Complexity}}:  \\
    There are memory and computational requirements for computing the query set $S_t$ at every query time $t$, as well as computing the final estimate $\hat{\theta}$. Specifically, adaptive selection of $S_t$ requires updating the posterior vector. There is also the computation complexity associated with the mapping $\gamma$ from $\.\pi(t)$ to the next query set ${S}_{t+1}$.
    \item \underline{\textit{The Query Geometry and the Query Cardinality}}: \\
    In many practical settings, the choice of the query set $S_t$ cannot be arbitrary. Let $\mathcal{A} \subseteq 2^{\{1,2,...,\frac{1}{\delta}\}}$ be the set of allowable query sets, i.e. consider $S_t \in \mathcal{A} \subsetneq 2^{\{1,2,...,\frac{1}{\delta}\}}$. We evaluate the algorithms in terms of the \rv{query geometry} of sets in $\mathcal{A}$. One practically relevant choice of $\mathcal{A}$, motivated by the visual search \cite{Naghshvar2013ISIT} and initial beam alignment (\cite{Chiu2019}) applications, is when $\mathcal{A}  = \mathcal{I} := \{ \{ i: a \leq i \leq b \} : 1\leq a< b \leq \frac{1}{\delta} \}$, i.e. when the query sets are constrained to be contiguous intervals. In such a case, we say that the search strategy is with a connected/contiguous query geometry. In fact, we will see that the connected query geometry of $\mathcal{I}$ offers an immediate reduction of computational and memory complexity in tracking the posterior.
    
    Furthermore, the cardinality of $\mathcal{A}$ determines the memory footprint of the algorithm. A smaller query cardinality is favorable for applications where the construction of the query set itself is non-trivial and a an offline construction with a codebook-based approach is preferable (e.g. the beam alignment problem in \cite{Chiu2019}). Hence, we characterize the query cardinality of the algorithms as the cardinality of $\mathcal{A}$.
\end{enumerate}

To get an understanding of the importance of the query geometry, let us present the reduction of computational and memory complexity just by the constraint of connected query geometry. We see this through the following lemma: 
\begin{lemma} \label{lemmaS:simple}
For connected query geometry $S_n \in \mathcal{I} := \{ \{ i: a \leq i \leq b \} : 1\leq a< b \leq \frac{1}{\delta} \}$, $n=1,2,...,t$ with uniform prior $\pi_i(0)=\delta$ for all $i$, the posterior at time $t$ can be written as a simple function with at most $2t+1$ intervals. Specifically, there exist a sequence of disjoint partition of \rv{ $[\frac{1}{\delta}] = \cup_{u=0}^{2t} J^{(t)}_u$, $J^{(t)}_u \in \mathcal{I}$ such that
\begin{equation} \label{eqS:simple}
    \pi_i(t) = \sum_{u=0}^{2t} \frac{\pi_{J^{(t)}_u}}{|J^{(t)}_u|} \mathds{1}_{J^{(t)}_u}(i), \ t = 1,2,...
\end{equation}
where $\pi_S := \sum_{i\in S} \pi_i$ for a set $S$.}
\end{lemma}

The proof of lemma \ref{lemmaS:simple} follows from Procedure \ref{proc:seqBayes}. An immediate and important implication is that the complexity of tracking the
posterior is of order $O(\tau)$ under the connected query geometry. \rv{ This is significant especially since for any algorithm with positive information rate (including all of our proposed algorithm as we will see) $\tau$ is of order $O(\log \frac{1}{\delta})$, implying the computational and memory complexity of $O(log \frac{1}{\delta})$}

\begin{procedure}[h!tb]  
\rv{\textbf{Notation}: $\mathcal{I}$ denotes a collection of intervals and $\.\pi_{\mathcal{I}} := (\pi_{J_1}, \pi_{J_2}, ... \pi{J_n})$ where $\mathcal{I}=\{J_1,...,J_n\}$}\;
\textbf{Input}: $(\.\pi_{\mathcal{I}^{(t)}}(t), \mathcal{I}^{(t)},S_{t+1},Y_{t+1})$ where $S_{t+1} = \{i: s_1 \leq i \leq s_2\}$\;
\textbf{Output}: $(\.\pi_{\mathcal{I}^{(t+1)}}(t+1), \mathcal{I}^{(t+1)})$ \;
\rv{Find $J^{(t)}_{t_1},J^{(t)}_{t_2}$ such that $s_1\in J^{(t)}_{t_1}$ and $s_2\in J^{(t)}_{t_2}$\;
\For{$0\leq u < t_1$}{
$J^{(t+1)}_u = J^{(t)}_u$, $\pi_{J^{(t+1)}_u}(t) = \pi_{J^{(t)}_u}(t)$ \;
}
$J^{(t+1)}_{t_1} =  [ \min J_{t_1}, s_1-1 ]$, $\pi_{J^{(t+1)}_{t_1}}(t) = \frac{|J^{(t+1)}_{t_1}|}{|J^{(t)}_{t_1}|} \pi_{J^{(t)}_{t_1}}(t)$ \;
$J^{(t+1)}_{t_1+1} =  [ s_1, \max J_{t_1} ] $, $\pi_{J^{(t+1)}_{t_1+1}}(t) = \frac{|J^{(t+1)}_{t_1+1}|}{|J^{(t)}_{t_1}|} \pi_{J^{(t)}_{t_1}}(t)$ \;
\For{$t_1 + 2 \leq u < t_2 + 1$}{
$J^{(t+1)}_u = J^{(t)}_{u-1}$, $\pi_{J^{(t+1)}_u}(t) = \pi_{J^{(t)}_{u-1}}(t)$ \;
}
$J^{(t+1)}_{t_2+1} = [\min J_{t_2}, s_2-1 ] $,  $\pi_{J^{(t+1)}_{t_2+1}}(t) = \frac{|J^{(t+1)}_{t_2+1}|}{|J^{(t)}_{t_2}|} \pi_{J^{(t)}_{t_2}}(t)$ \;
$J^{(t+1)}_{t_2+2} =  [ s_2, \max J_{t_2} ] $ , $\pi_{J^{(t+1)}_{t_2+2}}(t) = \frac{|J^{(t+1)}_{t_2+2}|}{|J^{(t)}_{t_2}|} \pi_{J^{(t)}_{t_2}}(t)$ \;
\For{$t_2 + 3 \leq u <\leq 2t+2$}{
$J^{(t+1)}_u = J^{(t)}_{u-2}$,  $\pi_{J^{(t+1)}_u}(t) = \pi_{J^{(t)}_{u-2}}(t)$ \;
}
\# Bayes' rule:
\begin{equation} \label{eqS:seqBayes}
\begin{aligned}
    &\pi_{J^{(t+1)}_u}(t+1) = \\
    &  \frac{ \pi_{J^{(t+1)}_{u}}(t)  \prob{Y_{t+1}| X_{t+1} = \mathds{1}( t_1 +1 \leq u \leq t_2 +1)  } }{\sum_{{u'}=0}^{2t+2}  \pi_{J^{(t+1)}_{u'}}(t)  \prob{Y_{t+1}| X_{t+1} = \mathds{1}(\theta \in J^{(t+1)}_{u'}) } }\;
\end{aligned}
\end{equation}}
\caption{1(Bayes' Rule with sequential binning)} 
\label{proc:seqBayes}
\end{procedure}


\section{Proposed Search Strategies and Main Results}
\label{secS:proposed_alg}

In this section, we introduce three proposed search strategies: sorted Posterior Matching ($sortPM$), dyadic Posterior Matching ($dyaPM$), and hierarchical Posterior Matching ($hiePM$). We give a summary of our main results in terms of rate (query time complexity), computational and memory complexity, as well as query geometry and cardinality in Table \ref{tabS:comp_all}. The three proposed algorithms are all variants of binary Posterior Matching proposed originally by Horstein \cite{Horstein63}, and later analyzed by \cite{Shayevitz11,Naghshvar2015,Li2014}. In this paper, we limit our attention to the deterministic variant of Posterior Matching which simply queries the bins to the left of posterior median. In other words,
\begin{equation}
    S_{t+1}^{(\text{PM})} = \gamma_{\text{PM}}(\.\pi_t) = \{i: i\leq k^*_{\text{PM}}\},
\end{equation}
where $k^*_{\text{PM}}$ is the bin index closet to the posterior median, $i.e.$ $k^*_{\text{PM}}=\argmin_k |\pi_{[1,k]}(t) - \frac{1}{2}|$.
\begin{remark}
By construction, $\prob{ \mathds{1}( \theta \in S_{t+1}) |Y_1^{t}} \approx \frac{1}{2}$. In other words, under Posterior Matching, $X_t$ has the desirable property of maximum (conditional) entropy. \rv{However, since the measurement noise, whose variance increases with the size, could be excessively large and so the information per query, $I( \mathds{1}( \theta \in S_{t+1}), Y_{t+1})$, can be significantly small. To address this, $sortPM$, which is proposed next, attempts to both ensure $\prob{ \mathds{1}( \theta \in S_{t+1}) |Y_1^{t}} \approx \frac{1}{2}$ and set $S_{t+1}$ to be of small cardinality.}
\end{remark}

{\renewcommand{\arraystretch}{1.3}
\begin{table*}[h!tb]  
    \centering
\begin{tabular}{ c| c |c | c| c |c | c }
& Asym. \# Query  & Reliability:  & Computation & Memory & Query & Query \\
& Rate: $\lim_{\delta}\frac{\log(1\slash\delta)}{\expect{\tau_{\epsilon,\delta}}} $ & $\lim_{\epsilon}\frac{\log(1\slash\epsilon)}{\expect{\tau_{\epsilon,\delta}}} $  & (each query) & Complexity & Geom. & Card. \\
\hline 
medianPM \cite{Horstein63}  & ${I(\frac{1}{2},p_{\max})} $ &${I(\frac{1}{2},p_{\max})} $ & $O(\log\frac{1}{\delta})$ &  $O(\log\frac{1}{\delta})$ & conn. & $O(\frac{1}{\delta})$ \\
\hline 
maxEJS \cite{Naghshvar2013ISIT}  & ${I(\frac{1}{2},p_{\min})}$ & $C_1(p[\delta])$  & $O(2^{\frac{1}{\delta}})$ & $O(\frac{1}{\delta})$ & disj. & $O(2^{\frac{1}{\delta}})$\\
\hline 
3rand \cite{Kaspi2018}  & ${I(\frac{1}{2},p_{\min})}$  & $C_1(p[\delta])$ & $O(\frac{1}{\delta})$\footnotemark[1] & $O(\frac{1}{\delta}) $  & disj. & $O(2^{\frac{1}{\delta}})$  \\
\hline 
$sortPM$  & ${I(\frac{1}{2},p_{\min})}$  [Thm\ref{thmS:sortPM}]  & $C_1(p[\delta])$ & $O(\frac{1}{\delta}\log\frac{1}{\delta})$ \footnotemark[2]  & $O(\frac{1}{\delta})$ & disj. & $O(2^{\frac{1}{\delta}})$\\
\hline 
$dyaPM$ & ${I(\frac{1}{2},p_{\min})}$ [Thm\ref{thmS:dyaPM}]  & $C_1(p[\delta])$ & $O(\log\frac{1}{\delta})$  & $O(\log\frac{1}{\delta})$ & conn. & $O((\frac{1}{\delta})^2)$\\
\hline 
$hiePM$ & ${I(\frac{1}{3},p_{\min})}$ [Thm\ref{thmS:hiePM}]  & $C_1(p[\delta])$ & $O(\log\frac{1}{\delta})$ & $O(\log\frac{1}{\delta})$ & conn. & $O(\frac{1}{\delta})$\\
\end{tabular}

\caption{Comparisons between different search strategies and Main results}
\rv{ where $C_1(\cdot)$ and $I(\cdot,\cdot)$ are defined in section \ref{sec:notation}, and $p_{\max}:=\max_S p[|S|] = p[1\slash 2]$ and $p_{\min}:=\min_S p[|S|] = p[\delta]$ are the worst and best possible noise model.}
    \label{tabS:comp_all}
\end{table*}
}
\footnotetext[2]{The complexity of random search strategy has to account for the complexity of the optimal decoder which requires the tracking of the posterior vector. \rv{We acknowledge that sub-optimal strategies can be utilized to trade-off memory and complexity.}}
\footnotetext[3]{A more sophisticated implementation via sequential quantization of $sortPM$ is possible , where both memory and computational complexity can be reduced to be of order $O((\log\frac{1}{\delta}) (\log\log\frac{1}{\delta}) )$}

\subsection{Sorted Posterior Matching} 


Under Sorted Posterior Matching ($sortPM$) strategy, the posterior matching step is preceded by a sorting operation on the posterior vector. In particular, consider the sorted posterior $\.\pi^{\downarrow}(t)$ obtained via the corresponding sorting operation $\sigma_t$: $\pi_{i}(t) \equiv \pi^{\downarrow}_{\sigma_t(i)} (t)$. Let $k^* = \argmin_{k} | \pi^{\downarrow}_{[1,k]}(t) - \frac{1}{2} |$. Under $sortPM$,
\begin{equation}
    S_{t+1} = \gamma_s(\.\pi(t)) =  \{  i: \sigma_t(i) \in [1,k^*] \},
\end{equation}
is queried.

\begin{algorithm}[h!tb] 
 \textbf{Input}: resolution $\frac{1}{\delta}$, error probability $\epsilon$, fixed stopping time $n$, \textit{stopping-criterion}\\
 \textbf{Output}: estimate of the target location $\hat{\theta}$ after $\tau$ queries\\
 \textbf{Initialization}: $\pi_i(0) = \delta$ for all $i=1,2,...,1\slash \delta$, \\
 \For{$t=0,1,...$ }{
    \# Design the search region by sorted posterior \\[-2mm]
    \begin{gather}
    \begin{aligned}
    k^* &= \argmin_{k} | \pi^{\downarrow}_{[1,k]}(t) - 1\slash 2 |\\
    S_{t+1} &= \gamma_s(\.\pi(t)) =  \{  i: \sigma_t(i) \in [1,k^*] \},
    \end{aligned}
    \end{gather}

    \setlength{\abovedisplayskip}{0pt} \setlength{\abovedisplayshortskip}{0pt}
    \# Take next measurement \\
     $Y_{t+1} = \mathds{1}(\theta\in S_{t+1}) \oplus Z_{t+1}$\\
    \# Posterior update by Bayes' Rule    \\
     $
         \.\pi(t+1) \leftarrow Y_{t+1}, \.\pi(t) 
     $ \\
    \# Stopping criteria \\
    \textit{case: stopping-criterion} = fixed length (FL)\\
    \If{ $t+1=n$}{
    break\;}
    
    \textit{case: stopping-criterion} = variable length (VL)\\
    \If{ $ \max_i  \pi_i(t+1) > 1-\epsilon$}{
    break\;}
}
$\tau = t+1$  (length of the search)\\
$\hat{\theta} =  \argmax_i \pi_i(\tau)$\\
\caption{Sorted Posterior Matching}
\label{algS:sortPM}
\end{algorithm}

\begin{theorem} \label{thmS:sortPM}
The expected search time of $sortPM$ of achieving resolution $\delta>0$ and reliability $0<\epsilon<1$ can be upper bounded by 
\begin{equation}
    E[\tau_{\epsilon,\delta} ] \leq   \frac{\log (1\slash \delta)}{I(1\slash 2, p[\alpha])} + \frac{\log (1\slash \epsilon)}{C_1(p[\delta])}  + \rv{ o(\log \frac{1}{\delta \epsilon}) },
\end{equation}
for any fixed $\alpha >  (e\log\frac{1}{\delta \epsilon})^{-K_s}$, where 
\rv{
\begin{equation}
    K_s  := 
\max \[\{  \frac{1}{2} D \[( \frac{1}{4} B_1 + \frac{3}{4} B_0 \Big\| B_0 \])  ,   \frac{1}{8} D\[( B_1 \Big \| \frac{3}{4} B_1 + \frac{1}{4} B_0 \]) \]\} > 0,
\end{equation}
$B_1 = \text{Bern}(1-p[1\slash 2])$, $B_0 = \text{Bern}(p[1\slash 2])$ and $D \[( \cdot \| \cdot \])$ is the KL divergence defined in Sec. \ref{sec:notation}.
}
\end{theorem}
\begin{proof}
See Appendix \ref{proofS:sortPM}.
\end{proof}

\begin{remark} \label{remarkS:sortPM_asym}
By first taking $\delta\rightarrow 0$ and then $\alpha \rightarrow 0$, Theorem~\ref{thmS:sortPM} together with the corresponding converse theorem [Theorem~1 in~\cite{Naghshvar2013ISIT}] implies that $sortPM$ achieves the best possible acquisition rate $I(1\slash 2, p_{\min})$ and the best reliability exponent $C_1(p[\delta])$ (by taking $\epsilon \rightarrow 0$).
\end{remark}

\begin{remark}
Even though $sortPM$, as well as prior works such as $maxEJS$ \cite{Naghshvar2013ISIT} and 3-phase random search \cite{Kaspi2018}, are asymptotically optimum in time complexity under measurement-dependent noise, they, in general, do not admit any constraint on the query set they choose. In other words, the query cardinality of these algorithms are of a prohibitive order $O(2^{\frac{1}{\delta}})$. Furthermore, the unconstrained query geometry prevents the applicability to many applications where connected query set or other specific geometry is preferred  (such as visual search \cite{Naghshvar2013ISIT}). $HiePM$, described next, restricts the query set.
\end{remark}

\subsection{Hierarchical Posterior Matching}
Motivated by many applications where disconnected inspecting regions is often infeasible, here we proposed a novel low-complexity search strategy which we call Hierarchical Posterior Matching, $hiePM$. $HiePM$ utilizes the hierarchical query geometry that is used in the noiseless binary search. For the brevity of presentation, we assume that $\frac{1}{\delta} = 2^L$ for some $L>0$. The hierarchical query geometry is therefore written as 
$\mathcal{H} = \{ H_l^m : l=0,1,2,..., m=0,1,2,...,2^l-1 \}$ where $H_l^m = \{ m 2^{L-l} +1 , m 2^{L-l} +2, ... , (m+1) 2^{L-l}  \}$. This query geometry, as shown in Fig. \ref{figS:dyaPM}, can be represented by a binary tree recursively by 
\begin{equation} \label{eqS:binaryTree}
    H^m_l = H^{2m}_{l+1} \cup H^{2m+1}_{l+1}, \ \ l=0,1,2,3,..., L.
\end{equation}

\begin{figure}
    \centering
    \includegraphics[width=0.8\textwidth]{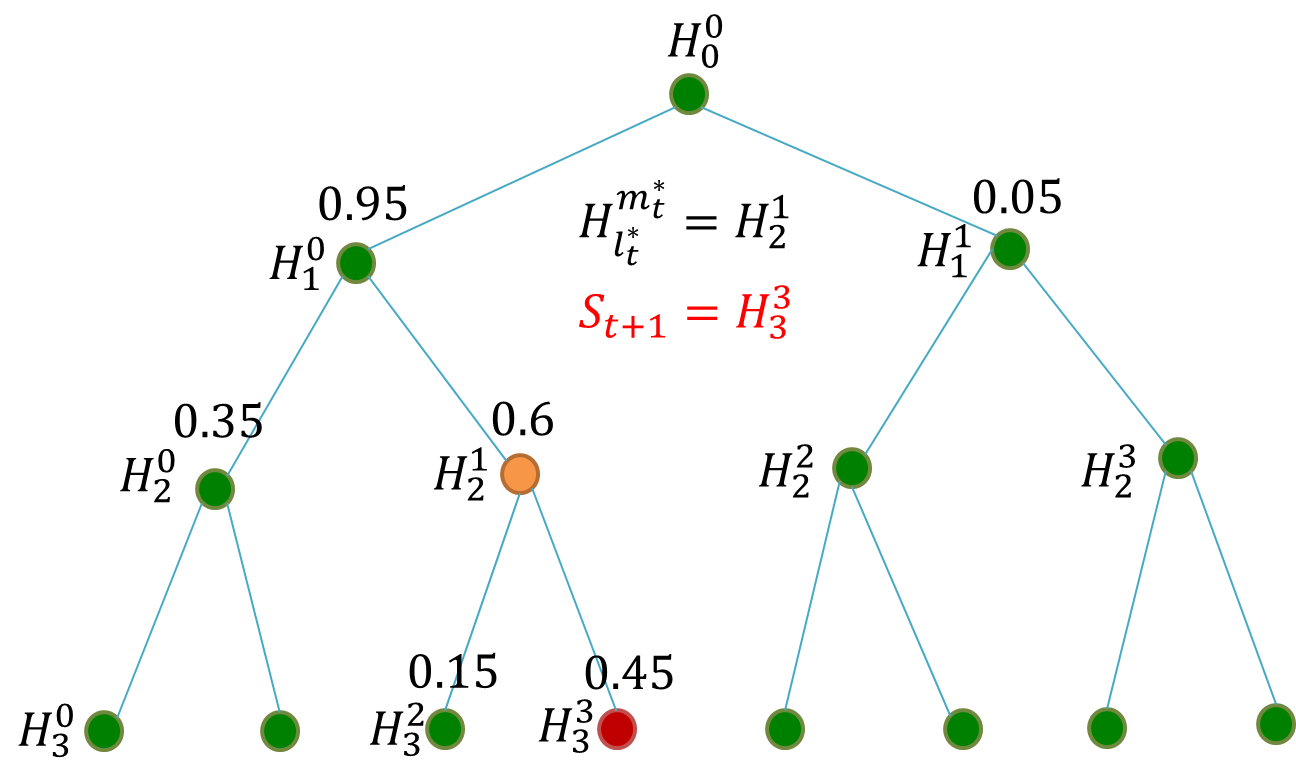}
    \caption{Binary search tree and the posterior for a given time $t$}
    \label{figS:dyaPM}
\end{figure}

By ensuring that $S_t \in \mathcal{H} \subseteq \mathcal{I}$, $hiePM$ ensures a connected query geometry. More precisely, $hiePM$ applies the Posterior Matching hierarchically along the binary tree as follows. Let
\begin{equation} \label{eqS:mass}
\begin{aligned}
    l^*_t &= \argmax_l \[\{ \max_m \pi_{H^m_l}(t) \geq \frac{1}{2}  \]\}, \\
    m^*_t &= \argmax_m \pi_{H^m_{l^*_t}}(t),
\end{aligned}
\end{equation}
and the hierarchical posterior matching with
\begin{equation} 
    \begin{aligned}
               (l_{t+1},m_{t+1}) = \argmin_{(l',m') \in \{(l_t^*,m^*_t),(l_t^*+1,2m^*_t),(l_t^*+1,2m^*_t+1)\}} \[|\pi_{H^{m'}_{l'}}(t) - \frac{1}{2}\]|.
    \end{aligned}
\end{equation}
In other words, $(l_{t+1},m_{t+1})$ identifies the node with the posterior closet to $\frac{1}{2}$. As such, querying $S_{t+1} = H_{l_{t+1}}^{m_{t+1}}$ ensures a high conditional entropy, while the size of the set is kept small to ensure near optimal time complexity: (See Algorithm~\ref{algS:hiePM} for more details on the construction of $hiePM$).

\begin{algorithm}[h!tb] 
 \textbf{Input}: resolution $\frac{1}{\delta}=2^L$, error probability $\epsilon$, fixed stopping time $n$, \textit{stopping-criterion}\\
 \textbf{Output}: estimate of the target location $\hat{\theta}$ after $\tau$ queries\\
\textbf{Initialization}: $\.\pi_{ \mathcal{I}^{(0)} }(0) = 1$, $\mathcal{I}^{(0)} = \{ (1,2,...,2^L) \} $ \\
 \For{$t=1,2,...$}{
    $l^*_t = \argmax_l \[\{ \max_m \pi_{H^m_l}(t) \geq \frac{1}{2}  \]\} $\;
    $m^*_t = \argmax_m \pi_{H^m_{l^*_t}}(t) $\;
    \# Match half posterior along the hierarchy $l$
        \begin{equation} 
    \begin{aligned}
               &(l_{t+1},m_{t+1}) = \\
               &\argmin_{(l',m') \in \{(l_t^*,m^*_t),(l_t^*+1,2m^*_t),(l_t^*+1,2m^*_t+1)\}} \[|\pi_{H^{m'}_{l'}}(t)  - \frac{1}{2}\]| ;
    \end{aligned}
\end{equation}
$S_{t+1} = H^{m_{t+1}}_{l_{t+1}}$\;
    \setlength{\abovedisplayskip}{0pt} \setlength{\abovedisplayshortskip}{0pt}
    \# Take next measurement \\
     $Y_{t+1} = \mathds{1}(\theta\in S_{t+1}) \oplus Z_{t+1}$\\
    \# Posterior update by Bayes' Rule (Procedure \ref{proc:seqBayes})   \\
$(\.\pi_{\mathcal{I}^{(t+1)}}(t+1), \mathcal{I}^{(t+1)}) \leftarrow (\.\pi_{\mathcal{I}^{(t)}}(t), \mathcal{I}^{(t)},S_{t+1},Y_{t+1})$ \;
    \# Stopping criteria \\
    \textit{case: stopping-criterion} = fixed length (FL)\\
    \If{ $t+1=n$}{
    break\;}
    
    \textit{case: stopping-criterion} = variable length (VL)\\
    \If{ $ \max_i  \pi_i(t+1) > 1-\epsilon$}{
    break\;}
}
$\tau = t+1$  (length of the search)\\
$\hat{\theta} = \argmax_i \pi_i(\tau)$\\
\caption{Hierarchical Posterior Matching}
\label{algS:hiePM}
\end{algorithm}

\begin{theorem} \label{thmS:hiePM}
The expected search time of $hiePM$ for achieving resolution $\delta>0$ and reliability $0<\epsilon<1$ can be upper bounded by 
\begin{equation}
    E[\tau_{\epsilon,\delta} ] \leq   \frac{\log (1\slash \delta)}{I(1\slash 3, p[2^{-l}])} + \frac{\log (1\slash \epsilon)}{C_1(p[\delta])}  + \rv{ o(\log \frac{1}{\delta \epsilon}) },
\end{equation}
for any fixed $l > 0$ such that $ 2^{-l} >  (e\log\frac{1}{\delta \epsilon})^{-K_h}$, where 
\rv{
\begin{equation}
K_h := \min \Bigg\{  I \Big(\frac{1}{3},p[\frac{1}{2}] \Big), \frac{2}{3} D\Big(\frac{1}{3} B_1 + \frac{2}{3} B_0  \Big\| B_0 \Big) \Bigg\} > 0,    
\end{equation}
$B_1 = \text{Bern}(1-p[1\slash 2])$, $B_0 = \text{Bern}(p[1\slash 2])$ and $D \[( \cdot \| \cdot \])$ is the KL divergence defined in Sec. \ref{sec:notation}.
}
\end{theorem}
\begin{proof}
See Appendix \ref{proofS:hiePM_dyaPM}.
\end{proof}
\begin{remark}
As shown in Algorithm \ref{algS:hiePM}, both the computational and memory complexity of $hiePM$ are dominated by tracking the posterior representation $\.\pi_{\mathcal{I}^{(t)}}, \mathcal{I}^{(t)}$ in Procedure \ref{proc:seqBayes}. This together with Theorem \ref{thmS:hiePM} implies that the computational and memory complexity is of order $O(\log \frac{1}{\delta})$.
\end{remark}

\begin{remark}
The hierarchical query geometry $\mathcal{H}$ not only is connected but also is of a hierarchical structure, which is suitable for the applications such as heavy hitter detection in networking \cite{Wang2018} (monitoring pre-fix IP addresses) and bit-wise coding \cite{Chiu2018}. Furthermore, the query cardinality is only $|\mathcal{H}| = O(\frac{1}{\delta})$, rendering $hiePM$ a great candidate for beamforming applications \cite{Chiu2019}, where the query set is stored in memory
\end{remark}

\begin{remark}
Taking $\epsilon \rightarrow 0$, we see that $hiePM$ achieves the best possible error exponent $C_1(p_{\min})$. On the other hand, the achievable acquisition rate of $hiePM$ by Theorem \ref{thmS:hiePM} is $I(1\slash 3, p_{\min})$ which is slightly smaller but close to the capacity $I(1\slash 2, p_{\min})$. This, we believe, is a byproduct of our analysis that loosely bounds the posterior distribution of $\mathcal{X}$ when we restrict the query area $S_{t+1}$ to $\mathcal{H}$.
\end{remark}

\subsection{Dyadic Posterior Matching}
By using the hierarchical query $\mathcal{H}$, $hiePM$ gives a solution that allows for constraints on the connectedness of query geometry. To ensure the optimality in time complexity, we proposed another low-complexity search strategy which we call dyadic Posterior Matching, $dyaPM$. 


By the same procedure as in $hiePM$, $dyaPM$ first finds the smallest binary interval that contains more than half posterior, i.e. $H^{m^*_t}_{l^*_t} = \{ m^*_t 2^{L-l^*_t} +1,m^*_t 2^{L-l^*_t} +2,..., (m^*_t +1) 2^{L-l^*_t} \}$ as in equation (\ref{eqS:mass}). The $dyaPM$ algorithm then applies the Posterior Matching within $H^{m^*_t}_{l^*_t}$ by potentially appending/exclusing additional bins:
\begin{equation}
    S_{t+1} = [m^*_t 2^{L-l^*_t} +1,k^*]
\end{equation}
where $k^* = \argmin_{k} | \pi_{[m^*_t 2^{L-l^*_t} +1,k]}(t) - 1\slash 2 |$ (The whole procedure of $dyaPM$ is summarized in Algorithm~\ref{algS:dyaPM}).

\begin{algorithm}[h!tb] 
 \textbf{Input}: resolution $\frac{1}{\delta}$, error probability $\epsilon$, fixed stopping time $n$, \textit{stopping-criterion}\\
 \textbf{Output}: estimate of the target location $\hat{\theta}$ after $\tau$ queries\\
 \textbf{Initialization}: $\.\pi_{ \mathcal{I}^{(0)} }(0) = 1$, $\mathcal{I}^{(0)} = \{ (1,2,...,2^L) \} $ \\
 \For{$t=1,2,...$}{
    $l^*_t = \argmax_l \[\{ \max_m \pi_{H^m_l}(t) \geq \frac{1}{2}  \]\} $\;
    $m^*_t = \argmax_m \pi_{H^m_{l^*_t}}(t) $\;
$k^* = \argmin_{k} | \pi_{[m^*_t 2^{L-l^*_t} +1,k]}(t) - 1\slash 2 |$\;
$S_{t+1} = [m^*_t 2^{L-l^*_t} +1,k^*]$\;

    \setlength{\abovedisplayskip}{0pt} \setlength{\abovedisplayshortskip}{0pt}
    \# Take next measurement \\
     $Y_{t+1} = \mathds{1}(\theta\in S_{t+1}) \oplus Z_{t+1}$\\
    \# Posterior update by Bayes' Rule (Procedure \ref{proc:seqBayes})   \\
$(\.\pi_{\mathcal{I}^{(t+1)}}(t+1), \mathcal{I}^{(t+1)}) \leftarrow (\.\pi_{\mathcal{I}^{(t)}}(t), \mathcal{I}^{(t)},S_{t+1},Y_{t+1})$ \;
    \# Stopping criteria \\
    \textit{case: stopping-criterion} = fixed length (FL)\\
    \If{ $t+1=n$}{
    break\;}
    
    \textit{case: stopping-criterion} = variable length (VL)\\
    \If{ $ \max_i  \pi_i(t+1) > 1-\epsilon$}{
    break\;}
}
$\tau = t+1$  (length of the search)\\
$\hat{\theta} = \argmax_i \pi_i(\tau)$\\
\caption{Dyadic Posterior Matching}
\label{algS:dyaPM}
\end{algorithm}


\begin{theorem} \label{thmS:dyaPM}
The expected search time of $dyaPM$ of achieving resolution $1\slash \delta$ and reliability $0<\epsilon<1$, can be upper bounded by 
\begin{equation}
    E[\tau_{\epsilon,\delta} ] \leq   \frac{\log (1\slash \delta)}{I(1\slash 2, p[2^{-l}])} + \frac{\log (1\slash \epsilon)}{C_1(p[\delta])}  + \rv{ o(\log \frac{1}{\delta \epsilon}) },
\end{equation}
for any fixed $l>0$ such that $2^{-l} >  (e\log\frac{1}{\delta \epsilon})^{-K_d}$, where 
\rv{
\begin{equation}
 K_d := \min \Bigg\{ \min_{\rho\in [0,1\slash 4]} \max \{ f(\rho),g(\rho)\} ,  \min_{\rho\in [1\slash 4,1\slash 2]} f(\rho) \frac{1}{4} D\[( \frac{1}{4}  B_1 + \frac{3}{4} B_0 \Big\|  B_0 \]) \Bigg\}  > 0,
\end{equation}
$B_1 = \text{Bern}(1-p[1\slash 2])$, $B_0 = \text{Bern}(p[1\slash 2])$, 
\begin{equation}
f(\rho) :=  \rho D\[( B_1 \big\|   (3 \slash 4) B_1 + (1 \slash 4) B_0  \]),    
\end{equation}
\begin{equation}
    \begin{aligned}
        g(\rho) := (1\slash 2 - \rho)  D\Big( (1-4\rho)B_1 + 4\rho B_0 \ \Big\|  (1\slash 2 + \rho)B_1 +  (1\slash 2 - \rho)B_0 \Big)
    \end{aligned},
\end{equation}
and $D \[( \cdot \| \cdot \])$ is the KL divergence defined in Sec. \ref{sec:notation}.
}
\end{theorem}
\begin{proof}
See Appendix \ref{proofS:hiePM_dyaPM}.
\end{proof}

\begin{remark}
By taking $\delta\rightarrow 0$ and then $l\rightarrow \infty$, we conclude that $dyaPM$ achieves the best possible acquisition rate $I(1\slash 2, p_{\min})$. And by taking $\epsilon \rightarrow 0$, $dyaPM$ achieves the best reliability exponent $C_1(\delta)$. To the best of our knowledge, $dyaPM$ is the first and the only known algorithm with connected query geometry with asymptotic optimal time complexity under measurement-dependent noise 
\end{remark}

\begin{remark}
As shown in Algorithm \ref{algS:dyaPM}, both the computational and memory complexity are again dominated by tracking the posterior representation $\.\pi_{\mathcal{I}^{(t)}}, \mathcal{I}^{(t)}$ in Procedure \ref{proc:seqBayes}. By Theorem \ref{thmS:dyaPM} we know that the computational and memory complexity of $dyaPM$ is of order $O(\log \frac{1}{\delta})$.
\end{remark}

\subsection{Asymptotic Results} \label{sec:asym}
We illustrate the asymptotic results from Theorems \ref{thmS:sortPM}-\ref{thmS:dyaPM} in Fig.~\ref{figS:RE} in terms of the achievable rate-reliability $(R,E)$ pair when the noise profile is given as $p[x] = 0.1 + 0.5x$, $0\leq x \leq \frac{1}{2}$. Note that the blue line not only represents the achievable $(R,E)$ pair of $sortPM$, $3rand$, $maxEJS$, and $dyaPM$, but also illustrates the converse as it is given in \cite{Kaspi2014}. \rv{We note the significant improvement of $hiePM$, even though suboptimal in general, over medianPM. We also compare our analysis to the analysis of the independent random walk (IRW) algorithm from \cite{Wang2018}. Specializing the steps of the analysis in \cite{Wang2018} to our setting, we derive the achievable rate of IRW. In particular, we aim to optimize a few constants in the proposed analysis of \cite{Wang2018}: minimization of a moment generating function in Chernoff bound, and numerically solving the number of samples at ``local test'' of IRW for ensuring probability of miss detection and false alarm less than a control parameter. We then further optimize the rate expression of their analysis over that control parameter to get the final achievable rate of IRW. As we can see in Fig.~\ref{figS:RE}, the theoretical guarantee of IRW is far from the optimal achievable rate-reliability.}


\begin{figure}[ht!b]
    \centering
    \includegraphics[width=0.8\textwidth]{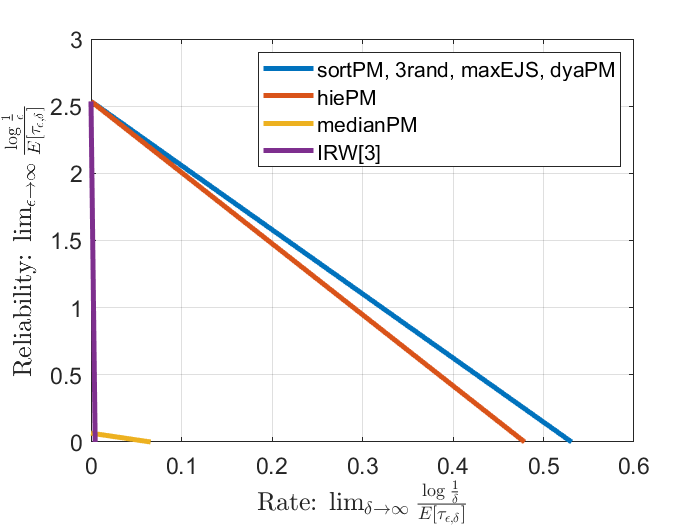}
    \caption{Achievable rate-reliability region:     The noise profile is set to be $p[x] = 0.1 + 0.5x $, $0\leq x \leq \frac{1}{2}$. (this means $p_{\min} = 0.1$, $p_{\max} = p[\frac{1}{2}] = 0.35$).}
    \label{figS:RE}
\end{figure}


\section{Numerical Examples}
\label{secS:num}
In this section, we give a detailed comparative analysis of our algorithm in a non-asymptotic sense. In particular, Fig. \ref{figS:err_linear} illustrates the error probability versus the number of queries at a fixed resolution $1\slash \delta=2^{15}$. We note that at this high resolution, the IRW algorithm from \cite{Wang2018} does not provide comparable performance due to sub-optimal asymptotic performance; hence we left out IRW from the Fig. \ref{figS:err_linear}. We compare all the algorithms under measurement-dependent Bernoulli noise where $p[x] = 0.1 + 0.5x$. For applications of our algorithm under non-Bernoulli noise profile, we refer readers to \cite{Chiu2019}. Note that for this noise profile, the acquisition capacity is $I(\frac{1}{2},p_{\min}=0.1) = 0.531$. In other words, our theory predicts a significant drop in error probability at $\tau > \frac{15}{0.531} = 28.25$. We also note that we expect this drop in error probability be significantly sharper for the variable length operation of the algorithm. We use algorithm-VL to represent the variable length termination of the algorithm. Likewise, we use algorithm-FL to represent the fixed length termination of the algorithm. 
\begin{figure}[ht!b]
    \centering
    \includegraphics[width=0.8\textwidth]{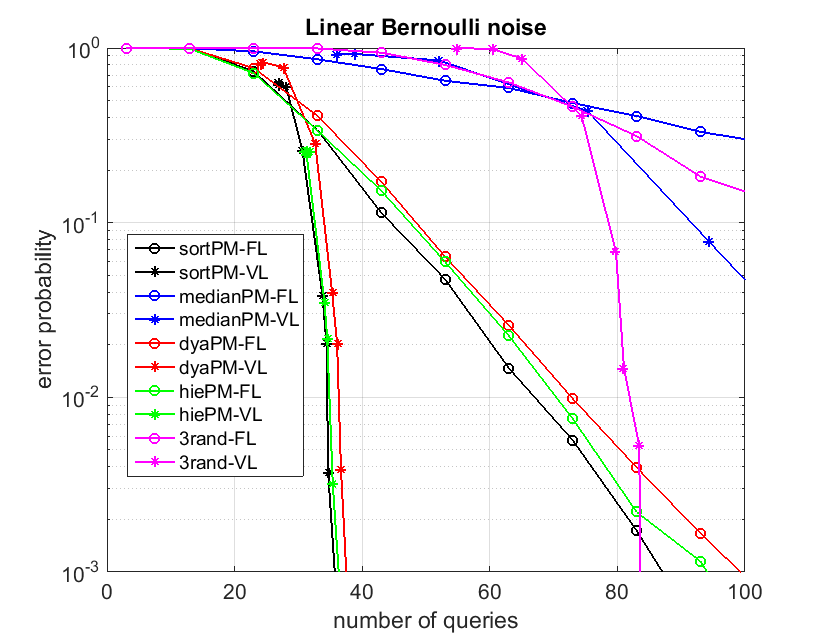}
    \caption{Error probability vs. number of queries: We consider the linear noise case where we set resolution $\frac{1}{\delta} = 2^{15}$ and Bernoulli noise with linear flipping probability $p[\delta|S|] = 0.1 + \delta|S|\slash 2 $ }
    \label{figS:err_linear}
\end{figure}

As we see in Fig. \ref{figS:err_linear}, the proposed algorithms $sortPM$, $dyaPM$, and $hiePM$ all enjoy the optimal error exponent $C_1(p[\delta])$ with variable length (VL) operation for both measurement independent and measurement-dependent noise, as predicted by Theorem 1-3. We also note that, despite the restriction of contiguous query area, $dyaPM$ and $hiePM$ perform almost the same as $sortPM$ both asymptotically and non-asymptotically in reliability. As expected, the classic PM performs rather poorly. On the other hand, while 3rand is also asymptotically optimal in reliability with VL operation, we note a non-negligible non-asymptotic performance drop compared to our proposed algorithms.

\section{Conclusion}

Our formulation of the four different complexities shows a systematic way that bridges theoretical studies of noisy search problem with practical engineering problem. Not only the low time/computational/memory complexity of the proposed strategies but also their \textit{query geometry} is shown to be suitable for practical applications. Particularly, restricting the query set with the hierarchical query geometry is found to be useful in the initial beam alignment problem in wireless communication \cite{Chiu2019}. Thanks to the Bayesian framework, our algorithms also adapts to different noise statistics (such as Poisson statistics in heavy hitter detection in networking), making our proposed algorithm potentially applicable in in many other target search applications. 

By the hierarchical query geometry, $hiePM$ also offers a natural generalization to a higher dimension or any structure that can be bisected. Applying $hiePM$ to more practical settings such as a target localization using drone \cite{Lu2018} is one of interesting extension of this paper. On the other hand, by Theorem \ref{thmS:hiePM}, we know that the (expected) number of queries grows only linearly in the number of dimensions. This benefit also renders $hiePM$ suitable for active learning problem where a learner tries to learn a classifier in multi-dimension by actively querying examples for labels.


%


\bibliographystyle{IEEEbib}
\bibliography{./refs}

\appendices

\section{\rv{Proof of the Theorems}}

\subsection{\rv{Overview of the proof}}
\label{proof-overview}
\rv{
The basic structure of the proof is by studying the crossing time of sub-martingales of relevant quantities that are related to the event of sufficient reliability in the target resolution. This idea of proof has been applied successfully on searching under measurement independent noise \cite{Shayevitz11, Naghshvar2015}. In our proof, we focus on the sub-martingales defined by the functional average of the log-likelihood of the posterior vector, whose expected drift is given by the Extrinsic Jensen Shannon (EJS) Divergence \cite{Naghshvar2015}. We review preliminary knowledge and facts about EJS in Appendix \ref{sec:reEJS}, from which we modify and derive lemmas (Lemma \ref{factS:EJSori_sort_dya} and Lemma \ref{factS:EJSori_sort_dya} stating the lower bound of the sub-martingale drift (EJS) under the proposed algorithms $sortPM$, $hiePM$, and $dyaPM$. The bounds in Lemma \ref{factS:EJSori_sort_dya} and Lemma \ref{factS:EJSori_sort_dya} are dependent on the noise statistics and hence on the size of the query set. Therefore, a naive application of these bounds on the crossing time study can only give the rate guarantee as if searching under the worst noise statistics (largest query set).

Our main contribution in the analysis of searching under size-dependent noise is to show that the size of the query set shrinks ``fast enough'' by using the corresponding search strategy. We show by Lemma \ref{lemma:total_prob} in Section \ref{sec:tot} (by the use of the total probability theorem) that it is sufficient (i.e. ``fast enough'') for establishing the asymptotic optimality if $\prob{\delta |S_t|>\alpha}$, the probability that the query set is larger than $any$ given fixed $\alpha>0$, decays exponentially over time. To prove such large deviation behaviour, we again analyze proper submartingales for each of the proposed search strategies. In particular, we prove that the sorted-and-coarse-binned average log-likelihood under $sortPM$ (Lemma \ref{lemmaS:EJSsortPM}), and the nested average log-likelihood under $hiePM$ and $dyaPM$ (Lemma \ref{lemmaS:EJShiePM}, and Lemma \ref{lemmaS:EJSdyaPM}) are all submartingales with positive drift, respectively. Together with the Azuma's Inequality (Lemma \ref{lemmaS:azuma}), we establish the aforementioned exponential decay probability, thereby concluding the assertions. We refer readers to below Appendix \ref{proofS:sortPM} and \ref{proofS:hiePM_dyaPM} for more details. 

Finally, we remark that the query set shrinkage is mainly based on the concentration behaviour of the posterior probability over time. This kind of concentration and error exponent of posterior matching has also been analyzed in the context of measurement-independent noise \cite{Henderson2013,Li2014}. However, the proof and the metrics in \cite{Henderson2013} and \cite{Li2014} are different from this paper and is not directly applicable here.  
}

\subsection{\rv{Upper-bounding the Expected Search Time with Measurement-Dependent Noise}} \label{sec:tot}
From the expected query time upper bound via the use of EJS (Fact \ref{factS:EJSthm}) and the query size $\delta |S_{t+1}|$ dependent lower bound of EJS given in Lemma \ref{factS:EJSori_sort_dya} and Lemma \ref{factS:EJSori_hiePM}, we see that intuitively we need $I(1\slash 2, p[\delta |S_{t+1}|])$ or $I(1\slash 3, p[\delta |S_{t+1}|])$ to be large, or equivalently the size of the query region $|S_{t+1}|$ to be small, in a certain sense. In particular, we can handle the query size shrinkage in a probabilistic manner by providing an exponentially decay tail. Indeed, we have the following proposition:
\begin{lemma} \label{lemma:total_prob}
Given any search strategy $\gamma$ with $\delta |S_{t+1}| \leq 1\slash 2$ and
\begin{align}
    &EJS(\.\pi(t),\gamma) \geq R(\delta|S_{t+1}|), \ \forall \ t\\
    &EJS(\.\pi(t),\gamma) \geq \tilde{\.\pi} E ,\ \forall \max_i \.\pi_i \geq  \tilde{\pi},
\end{align}
for some $R(\delta |S_{t+1}|)>0$ increasing in $\delta|S_{t+1}|$ and $E>0$. If further
\begin{equation}
    \prob{\delta|S_{t+1}| > \alpha} \leq k_0 e^{-t E_0}, \ \forall \ t > T_0
\end{equation}
for some $1\slash 2 >\alpha>\delta$, $k_0>0$, $E_0>0$, and $T_0 > \lceil \log\log (\frac{1}{\delta \epsilon}) \rceil $, the expected time of the strategy $\gamma$ achieving resolution $1\slash \delta$ and reliability $\epsilon$ can be upper bounded by
\begin{equation}
    E[\tau_{\epsilon,\delta} ] \leq   \frac{\log (1\slash \delta)}{R(\alpha)} + \frac{\log (1\slash \epsilon)}{E}  + g_{R,E}(\epsilon,\delta),
\end{equation}
where 
\begin{equation}
\begin{aligned}
     & g_{R,E}(\epsilon,\delta) := \frac{k_0 e^{-E_0}}{(1-e^{-e_0})  (\log \frac{1}{\delta \epsilon})^{E_0}}  \times \\
    & \qquad  \[( \lceil \log\log \frac{1}{\delta \epsilon}\rceil + \frac{\log \frac{1}{\delta} }{R(1\slash 2)} + \frac{\log \frac{1}{\epsilon}}{E} + f_{R(1\slash 2),E}(\epsilon,\delta) \]) \\
    & + \frac{k_0 e^{-2E_0}}{(1-e^{-e_0})^2(\log \frac{1}{\delta \epsilon})^{2E_0}} + \lceil \log\log \frac{1}{\delta \epsilon}\rceil  + f_{R(\alpha),E}(\epsilon,\delta)
\end{aligned}
\end{equation}
is of $o(\frac{1}{\delta \epsilon})$ as $\delta \rightarrow 0$ or $\epsilon \rightarrow 0$.
\end{lemma}
\begin{proof}
We prove this proposition via the total probability theorem and the re-start of the time homogeneous Markov chain $\.\pi(t)$. Specifically, let us define the ``bad" event $E_t = \{ \delta |S_{t+1}| > \alpha\}$  and the ``good" event $F_n = \bigcup_{t=n}^{\infty} E_t$. For every $n$, by total probability theorem and the union bound, we have 
\begin{equation} \label{eqS:total_prob}
    \begin{aligned}
       \expect{\tau_{\epsilon,\delta}} &= \int_{\Omega} \tau_{\epsilon,\delta} \ d\mathbb{P} \leq \sum_{t=n}^{\infty} \int_{E_t} \tau_{\epsilon,\delta} \ d\mathbb{P} + \int_{F_n^C} \tau_{\epsilon,\delta} \ d\mathbb{P} \\
       &= \sum_{t=n}^{\infty} \int_{E_t} \expect{ \tau_{\epsilon,\delta} |  \.\pi(t) } \ d\mathbb{P} + \int_{F_n^C} \tau_{\epsilon,\delta} \ d\mathbb{P}\\
        &\overset{(a)}{\leq} \sum_{t=n}^{\infty} \prob{E_t}  \[( t+  \frac{\log \frac{1}{\delta} }{R(1\slash 2)} + \frac{\log \frac{1}{\epsilon}}{E} + f_{R(1\slash 2),E}(\epsilon,\delta)  \]) \\
        & \qquad + \int_{F_n^C} \tau_{\epsilon,\delta} \ d\mathbb{P} \\
        &\overset{(b)}{\leq} \sum_{t=n}^{\infty} \prob{E_t}  \[( t+  \frac{\log \frac{1}{\delta} }{R(1\slash 2)} + \frac{\log \frac{1}{\epsilon}}{E} + f_{R(1\slash 2),E}(\epsilon,\delta)  \]) \\
        & \qquad + n + \frac{\log \frac{1}{\delta} }{R(\alpha)} + \frac{\log \frac{1}{\epsilon}}{E} + f_{R(\alpha),E}(\epsilon,\delta) ,
    \end{aligned}
\end{equation}
where $ f_{R,E}(\epsilon,\delta)$ is as defined in Fact \ref{factS:EJSthm}. Here $(a)$ follows from the time homogeneity of the Markov Chain $\.\pi(t)$ re-starting at time $t$, together with Fact~\ref{factS:EJSthm} and $\delta |S_{t+1}|\leq 1\slash 2$, written as
\begin{equation}
\begin{aligned}
    \expect{ \tau_{\epsilon,\delta} | \.\pi(t)} \leq t +  \frac{\log \frac{1}{\delta} }{R(1\slash 2)} + \frac{\log \frac{1}{\epsilon}}{E} + f_{R(1\slash 2),E}(\epsilon,\delta).
\end{aligned}
\end{equation}
Similar argument can be made for $(b)$ with $\delta |S_{t+1}|\leq \alpha$ for $t\geq n$ under event $F_n^C$. Now, plugging the assumption $\prob{E_t} = \prob{\delta |S_{t+1}| > \alpha} \leq k_0 e^{-t E_0}$
into (\ref{eqS:total_prob}) with some algebra, we have
\begin{equation}
\begin{aligned}
    &E[\tau_{\epsilon,\delta} ]  \leq  \frac{k_0 e^{-n E_0}}{1-e^{-E_0}} \times \\
    & \quad \[( n + \frac{e^{-nE_0}}{1-e^{-E_0}} +  \frac{\log \frac{1}{\delta}}{R(1\slash 2)} + \frac{\log \frac{1}{\epsilon}}{E}  + f_{R(1\slash 2),E}(\epsilon,\delta) \]) \\
    & \qquad + n + \frac{\log \frac{1}{\delta}}{R(\alpha)} + \frac{\log \frac{1}{\epsilon} }{E}  + f_{R(\alpha),E}(\epsilon,\delta).
\end{aligned}
\end{equation}
Letting $n= \lceil \log\log \frac{1}{\delta \epsilon} \rceil $, we have the assertion of the proposition.
\end{proof}

By lemma \ref{lemma:total_prob}, we can see that for proving Theorem \ref{thmS:sortPM},\ref{thmS:dyaPM},\ref{thmS:hiePM}, it is sufficient to provide exponential decay tail probability of a large query size $\prob{\delta |S_{t+1}| > \alpha}$ for each of the proposed algorithm $S_{t+1} = \gamma(\.\pi(t))$. The main idea of studying the event $\{\delta |S_{t+1}| > \alpha\}$ is to group the region into coarse bins of size $\alpha$ according to each of the search strategy. And by the nature of each algorithm the event $\{\delta |S_{t+1}| > \alpha\}$ is equivalent to the event that one coarse bin has posterior larger than half. By further considering a similar submartingale of an average log-likelihood as in (\ref{eqS:aver_log_like}) but over the coarse bin posterior, the problem is then transformed to be the tail probability of a level crossing of a strictly positively drifted submartingle, where we can bound it by the Azuma's Inequality (Lemma \ref{lemmaS:azuma}). Now let us provide the details:

\subsection{Proof of Theorem~\ref{thmS:sortPM}} \label{proofS:sortPM}
Along with the operation of $sortPM$, we first sort the posterior, and then group into bins with size $\delta |\text{bin}(q)| = \alpha$, written as
\begin{equation} \label{eqS:alpha_pos}
   \pi_{q}^{\alpha}(t) := \sum_{i\in \text{bin}(q)} \pi^{\downarrow}_i(t), \ q = 1,2,..., 1\slash \alpha ,
\end{equation}
where $\pi^{\downarrow}$ is the sorted posterior, $\text{bin}(q):= \{\frac{\alpha}{\delta} (q-1) +1,\frac{\alpha}{\delta} (q-1) +2, ... ,\frac{\alpha}{\delta} q \}$. For notational simplicity, we deal with the case where $1\slash \alpha$ and $\alpha \slash \delta$ are both integer (the proof follows similarly for non-integer case). Let us further define the average log-likelihood of the binned sorted posterior
\begin{equation} \label{eqS:alpha_U}
\begin{aligned}
     U_{\alpha}(t) &:= U\[(\.\pi^{\alpha}(t)\]) \\
     &= \sum_{q=1}^{1\slash\alpha} \pi_{q}^{\alpha}(t) \log \frac{\pi_{q}^{\alpha}(t)}{1-\pi_{q}^{\alpha}(t)}.
\end{aligned}
\end{equation}
By the query set selection rule in Algorithm \ref{algS:sortPM} together with the definition of $U_{\alpha}(t)$, under $sortPM$ strategy we have 
\begin{equation}
\begin{aligned}
\prob{ \delta |S_{t+1}| > \alpha } &\leq \prob*{ \pi_1^{ \alpha}(t) < \frac{1}{2} } \\
&\leq  \prob*{ U_\alpha(t) < 0 }.
\end{aligned}
\end{equation}
Now, by Lemma \ref{factS:EJSdiff} and Lemma \ref{lemmaS:EJSsortPM}, $U_\alpha(t)$ is a submartingale with bound difference 
\begin{equation}
\begin{aligned}
    &|U_\alpha(t+1) - U_\alpha(t)| \\
    &\leq B_\alpha: = \log(1\slash \alpha) + \log \frac{1-p_\text{min}}{p_\text{min}} + e.
\end{aligned}
\end{equation}
Further note that $U_\alpha(0) = -\log(1\slash \alpha-1) < -\log(1\slash \alpha)$ and together with lemma \ref{lemmaS:azuma}, we have
\begin{equation} \label{eqS:sortPM_exp_tail}
    \begin{aligned}
        \prob{ \delta |S_{t+1}| > \alpha } &\leq \prob*{ U_\alpha(t) < 0 } \\
        &\leq k_s e^{-t \frac{K_s^2}{2 (B_\alpha + K_s)^2} }\ \ \forall t > \frac{\log (\frac{1}{\alpha})}{K_s},
    \end{aligned}
\end{equation}
where $k_s = e^{  \frac{K_s \log(1\slash \alpha)}{K_s +  B_\alpha} }$. Since $\alpha > (e\log \frac{1}{\delta \epsilon})^{-K_s}$ and therefore $ \frac{\log(1\slash\alpha)}{K_s} < \lceil \log\log  \frac{1}{\delta \epsilon} \rceil$, by proposition \ref{lemma:total_prob}, we conclude the assertion.

\subsection{Proof of Theorem~\ref{thmS:hiePM} and \ref{thmS:dyaPM}}
\label{proofS:hiePM_dyaPM}
Without loss of generality, we may assume that the resolution $\delta$ is such that $L=\log_2(1\slash \delta)$ is an integer. If otherwise, we can choose a smaller $\delta'$ such that $\log_2(1\slash \delta')$ is an integer and the analysis will follow similarly without affecting the asymptotic conclusions. One of the key attribute of $dyaPM$ and $hiePM$ is the nested resolution due to the natural bisection. To analyze it, we introduce the posterior vector $\.\pi^{\{l\}} (t)$ of a nested resolution level $l < L$ with length $2^l$ where its elements are defined as
\begin{equation} \label{eqS:nestedpi}
    \pi_q^{\{l\}} (t) := \sum_{i\in \text{bin}(q)} \pi_i(t), \quad q = 1,2,...,2^l,
\end{equation}
where $\text{bin}(q) := \{ (q-1) 2^{L-l}+1,(q-1) 2^{L-l}+2,...,q2^{L-l}\}$. Further, we can also define the average log-likelihood on $\.\pi^{\{l\}}$ as
\begin{equation} \label{eqS:nestedU}
    U^{\{l\}}(t)  := \sum_{q=1}^{2^l}\pi_q^{\{l\}}(t)  \log \frac{\pi_q^{\{l\}}(t)}{1-\pi_q^{\{l\}}(t)}.
\end{equation}
We have
\begin{equation}
\begin{aligned}
\prob{ \delta |S_{t+1}| > 2^{-l} } &\leq \prob*{ \max_q \pi_q^{\{l\}}(t) < \frac{1}{2} } \\
&\leq  \prob*{ U^{\{l\}}(t) < 0 }.
\end{aligned}
\end{equation}
The proof then follows similarly as in the proof of Theorem \ref{thmS:sortPM}: Applying proposition \ref{lemma:total_prob} with $\alpha=2^{-l}$, where the corresponding submartingale properties of $U^{\{l\}}(t)$ is by Lemma \ref{lemmaS:EJSdyaPM} and Lemma \ref{lemmaS:EJShiePM} for $dyaPM$ and $hiePM$, respectively, hence we omitted the rest.

\section{\rv{Preliminaries: Average Log-Likelihood and the Extrinsic Jensen-Shannon Divergence}}
\label{sec:reEJS}

In this subsection, we review some useful concepts in~\cite{Naghshvar2015}. The average log-likelihood of the posterior is defined as
\begin{equation} \label{eqS:aver_log_like}
    U(t) \equiv U(\.\pi(t)) := \sum_{i=1}^{1 \slash \delta } \pi_i(t) \log \frac{\pi_i(t)}{1-\pi_i(t)},
\end{equation}
with the following property:
\begin{enumerate}
    \item  $U(t)$ is a submartingale with drift $EJS$.
    \begin{equation}
        \expect{ U(t+1) | \.\pi(t) } =  U(t) + EJS(\.\pi(t), \gamma),
    \end{equation}
    where $EJS$ is the Extrinsic Jensen-Shannon divergence, defined as
    \begin{equation}
        EJS(\.\pi(t), \gamma) = \sum_{i=1}^{1 \slash \delta } \pi_i(t) D\[( P_{y_t|  i , S_{t+1} } \Big\| P_{y_{t+1}| \neq i, S_{t+1}} \])
    \end{equation}
    with 
    \begin{equation} \label{eqS:Pyt_original}
    \begin{aligned}
        P_{y_{t+1}|  i , S_{t+1} } &:= \prob{ Y_{t+1}=y_{t+1} |{\theta} = i; S_{t+1} = \gamma(\.\pi(t))} \\
                        &= \prob{ Y_{t+1}=y_{t+1} | X_{t+1} =\mathds{1}(i\in S_{t+1})  }
    \end{aligned}
    \end{equation} and 
    \begin{equation}
        \begin{aligned}
            P_{y_{t+1}|\neq i, S_{t+1} } 
            &:= \prob{ Y_{t+1}=y_{t+1} |{\theta} \neq i; S_{t+1}} \\
            &= \sum_{j\neq i} \frac{\pi_j(t)}{1-\pi_i(t)} P_{y_{t+1}| j , S_{t+1}}.
        \end{aligned}
    \end{equation}
    \item Initial value $U(0) = -\log (\frac{1}{\delta}-1)$ is directly related to the logarithm of resolution and hence the targeting rate
    \item Level crossing of $U$ is directly related to the error probability, since 
    $
        \pi_i(t)<1-\epsilon\  \forall i  \ \Rightarrow\  U(t)< \log \frac{1- \epsilon}{\epsilon}.
    $
\end{enumerate}
Analyzing the random drift from time 0 with the initial value $U(0)$ up to the first crossing time $\nu:= \min\{ t: U(t) \geq \log \frac{1}{\epsilon} \}$ is closely related to the expected drift given by $EJS$. In particular, we can then establish an upper bound for the expected targeting time $\expect{\tau_{\epsilon,\delta}}$ in terms of the predefined error probability $\epsilon$ and the resolution $\delta$.  Specifically we have the following theorem:

\begin{fact}[Theorem 1 in~\cite{Naghshvar2015}] \label{factS:EJSthm}
Define 
\begin{equation}
    \tilde{\pi} := 1- \frac{1}{1+\max \{ \log (1\slash \delta), \log (1\slash \epsilon) \}}.
\end{equation}
For adaptive search strategy with query region $S_t$, if \begin{equation}
    EJS(\.\pi(t), \gamma) \geq R \quad \forall t \geq 0
\end{equation}
and 
\begin{equation}
    EJS(\.\pi(t), \gamma) \geq \tilde{\pi} E \quad \forall  t \geq 0 \text{ s.t. } \max_{i} \pi_i(t) \geq \tilde{\pi},
\end{equation}
we have the expected targeting time associated with error probability $\epsilon$ and resolution $\delta$ bounded by
\begin{equation}
    E[\tau_{\epsilon,\delta} ] \leq   \frac{\log (1\slash \delta)}{R} + \frac{\log (1\slash \epsilon)}{E}  + f_{R,E}(\epsilon,\delta) 
\end{equation}
where $ f_{R,E}(\epsilon,\delta) = \frac{\log\log \frac{1}{\delta \epsilon}}{R} + \frac{1}{E} + \frac{96}{RE} (\frac{1-p[\delta]}{p[\delta]})^2 $.
\end{fact}
\begin{proof}
The proof of Fact~\ref{factS:EJSthm} follows similarly the proof of [Theorem 1,~\cite{Naghshvar2015}]. 
\end{proof}
\begin{fact}[Lemma 2 in \cite{Naghshvar2015}] \label{factS:JS_EJS}
The EJS divergence is lower bounded by the Jensen Shanon (JS) divergence : 
\begin{equation}
    EJS(\.\pi(t), \gamma) \geq JS(\.\pi(t), \gamma),
\end{equation}
where
    \begin{equation}
        JS(\.\pi(t), \gamma) = \sum_{i=1}^{1 \slash \delta } \pi_i(t) D\[( P_{y_t|  i , S_{t+1} } \Big\| P_{y_{t+1}| S_{t+1}} \])
    \end{equation}
with $P_{y_{t+1}| S_{t+1}} := \sum_{i} \pi_(t) P_{y_{t+1}| i , S_{t+1}}.$ 
\end{fact}

\begin{lemma}[\rv{Theorem 1 in~\cite{Chiu2018}\footnote{This is one of the conference proceedings of this work}}] \label{factS:EJSori_hiePM}
Using the search strategy $hiePM$ with resolution $1\slash \delta$ and reliability $\epsilon$ on codebook $\mathcal{W}^L$ with $L = \log_2 (1\slash\delta)$, we have
    \begin{align}
        EJS(\.\pi(t),\gamma_h) &\geq I( 1\slash 3 , p[\delta|D_{l_{t+1}}|]) , \ \forall \ t  \label{eqS:hiePM_Rlb} \\
        EJS(\.\pi(t),\gamma_h) &\geq \tilde{\.\pi} C_1(p_{\text{min}}) ,\ \forall \max_i \.\pi_i \geq  \tilde{\pi}  \label{eqS:hiePM_Elb},
    \end{align}
where $\tilde{\pi} := 1- \frac{1}{1+\max \{ \log (1\slash \delta), \log (1\slash \epsilon) \}}.$
\end{lemma}
\begin{proof}
The proof of Lemma~\ref{factS:EJSori_hiePM} is a modification from proof of [Proposition 3,~\cite{Naghshvar2015}] by using Fact \ref{factS:JS_EJS}. We first prove equation (\ref{eqS:hiePM_Elb}). By the selection rule of $hiePM$, the last level codebook $S_{t+1} = D^{(l_{t+1}=\log_2(\frac{1}{\delta}))}$ is used whenever $\max_{i} \pi_i(t) \geq \tilde{\pi} > 1\slash 2$. Therefore, 
\begin{equation}
\begin{aligned}
     EJS(\.\pi(t), \gamma_h) &= \sum_{i=1}^{1\slash \delta} \pi_i(t) D\[( P_{\hat{y}_{t+1}|  i , \gamma_h} \Big\| P_{\hat{y}_{t+1}| \neq i, \gamma_h } \])\\
     &\geq \tilde{\pi} D\[( P_{\hat{y}_{t+1}| i , \gamma_h} \Big\| P_{\hat{y}_{t+1}| \neq i, \gamma_h} \]) \\
     &= \tilde{\pi}  D( \text{Bern}(1-p[S]) \| \text{Bern}(p[S]))  \\
     &= \tilde{\pi} C_1(p[\log_2 (1\slash\delta)  ]).
\end{aligned}
\end{equation}
It remains to show equation (\ref{eqS:hiePM_Rlb}). For notational simplicity, let 
\begin{equation}
    \rho \equiv \pi_{D^{l_{t+1}}}(t) := \sum_{i\in D^{l_{t+1}}} \pi_i(t)
\end{equation}
and $B^0 \equiv \text{Bern}(p[l_{t+1}])$, $B^1 \equiv \text{Bern}(1-p[l_{t+1}])$. We separate the proof into two cases:

If $\rho >  2 \slash 3$, we know that $l_{t+1} = \log_2(\frac{1}{\delta})$ by the selection rule of $hiePM$. Therefore, the set $D^{l_{t+1}}$ is of the smallest size 1. Let $ D^{l_{t+1}} = \{\ i_{t+1} \} $, we have
\begin{equation}
\begin{aligned}
     &EJS(\.\pi(t), \gamma_h) =\sum_{i=1}^{1\slash \delta} \pi_i(t) D\[( P_{\hat{y}_{t+1}|  i , \gamma_h} \Big\| P_{\hat{y}_{t+1}| \neq i, \gamma_h } \])\\
    & =   \rho D\big( B^1 \big\| B^0 \big)   \\ 
     & + \sum_{i\neq i_{t+1}}  \pi_{i}(t)  D\[( B^0 \Big\| \frac{ \rho }{1-\pi_{i}(t)}   B^1 + \frac{1-\rho-\pi_i(t)}{1-\pi_i(t)} B^0 \]) \\[2mm]
     &\stackrel{(a)}{\geq}  D\[( B^0 \Big\| \frac{1}{2}  B^1 + \frac{1}{2} B^0 \])  \\
     & =  I(1\slash 2, p[l_{t+1}] ) \geq I(1\slash 3, p[l_{t+1}] ),  \\ 
\end{aligned}
\end{equation}
where (a) is by the fact that $ D( B^1 \| B^0 )  =  D( B^0 \| B^1 )   $ and that 
$D(  B^0  \| \alpha  B^1 + (1-\alpha) B^0 )$ is increasing in $\alpha$ for $0\leq \alpha \leq 1$, together with $\frac{ \rho }{1-\pi_{i}(t)} >  2\slash 3> 1\slash 2$.

For the other case where $\rho \leq 2 \slash 3$, again by the selection rule of $hiePM$, we have $1 \slash 3 \leq \rho \leq 2 \slash 3$. Now we can lower bound the $EJS$ as
\begin{equation}
\begin{aligned}
     EJS(\.\pi(t), \gamma_h) &\stackrel{(a)}{\geq} JS(\.\pi(t), \gamma_h)\\
     & = \rho D \[( B^1  \Big\|  \rho B^0 + (1-\rho) B^1  \]) \\
     & \qquad + (1-\rho) D \[( B^0  \Big\|  \rho B^0 + (1-\rho) B^1  \])   \\
     &= I(\rho, p[l_{t+1}] ) \stackrel{(b)}{\geq} I( 1\slash 3 , p[l_{t+1}] )
\end{aligned}
\end{equation}
where (a) is by Fact \ref{factS:JS_EJS} and (b) is by the concavity of the mutual information with respect to the input distribution, the symmetric of $I(\rho, p[l_{t+1}] )$ around $\rho = 1\slash 2$ for symmetric channels, and together with $1\slash 3 \leq \rho \leq 2\slash 3$. This concludes the assertion.

\end{proof}

\begin{lemma}[\rv{Fact 1 in~\cite{Chiu2016}\footnote{This is one of the conference proceedings of this work}}] \label{factS:EJSori_sort_dya}
For both search strategies $sortPM$ and $dyaPM$ with resolution $1\slash \delta$ and reliability $\epsilon$, we have
\begin{align}
    &EJS(\.\pi(t),\gamma) \geq I(1\slash 2, p[\delta |S_{t+1}|]), \ \forall \ t\\
    &EJS(\.\pi(t),\gamma) \geq \tilde{\.\pi} C_1(p[\delta ]) ,\ \forall \max_i \.\pi_i \geq  \tilde{\pi} ,
\end{align}
where $\tilde{\pi} := 1- \frac{1}{1+\max \{ \log (1\slash \delta), \log (1\slash \epsilon) \}}.$
\end{lemma}
\begin{proof}
The proof of Lemma~\ref{factS:EJSori_sort_dya} follows along the proof of [Proposition 3,~\cite{Naghshvar2015}].
\end{proof}

\section{\rv{Technical Lemmas}}

\begin{lemma}\label{factS:EJSdiff}
    The absolute difference between $U(\.\pi(t+1))$ and $U(\.\pi(t))$ is bounded by the entropy of $\.\pi(t)$, written as
    \begin{equation}
        |U(\.\pi(t+1)) - U(\.\pi(t))| \leq \log \frac{1-p_\text{min}}{p_\text{min}} + H(\.\pi(t)) +  e,
    \end{equation}
\end{lemma}
\begin{proof}
\begin{equation}
\begin{aligned}
    &|U(\.\pi(t+1)) - U(\.\pi(t))| \\
    &= \sum_{i=1}^{1\slash \delta} \pi_i(t+1) \log \frac{\pi_i(t+1)}{1-\pi_i(t+1)} - \sum_{i=1}^{1\slash \delta} \pi_i(t) \log \frac{\pi_i(t)}{1-\pi_i(t)} \\
                  &\leq   \sum_{i=1}^{1\slash \delta} \pi_i(t+1)  \[| \log \frac{\pi_i(t+1)}{1-\pi_i(t+1)} - \log \frac{\pi_i(t)}{1-\pi_i(t)}   \]| \\
                  & \qquad + \sum_{i=1}^{1\slash \delta} | \pi_i(t+1) -  \pi_i(t) |  \[|  \log  \frac{\pi_i(t)}{1-\pi_i(t)}  \]|\\
                  &  \stackrel{(a)}{\leq}  \log \frac{1-p_\text{min}}{p_\text{min}} + \sum_{\pi_i(i)<\frac{1}{2}} \pi_i(t) (1-\pi_i(t)) \log \frac{1-\pi_i(t)}{\pi_i(t)} \\
                  & \qquad +  \sum_{\pi_i(i)\geq \frac{1}{2}} \pi_i(t) (1-\pi_i(t)) \log \frac{\pi_i(t)}{1-\pi_i(t)} \\
                  & \leq \log \frac{1-p_\text{min}}{p_\text{min}} + H(\.\pi(t)) + \max_x x \log \frac{1}{x},
\end{aligned}
\end{equation}
where (a) is by lemma 6 in \cite{Naghshvar2015} and that $p[|S_t|] \geq p_\text{min}$.
\end{proof}

\begin{lemma}[Azuma's Inequality, Lemma 1 in~\cite{Chiu2018}\footnote{This is one of the conference proceedings of this work}] \label{lemmaS:azuma}
Given a submartingale $U(t)$ with $U(0)<0$ with respect to another random process $\.\pi(t)$. If $U(t)$ has bounded difference, $i.e.$ $|U(t+1)-U(t)|<B$ for some $B \in \mathbb{R}^+$, and that the expected difference is strictly positive, i.e. 
\begin{equation}
    \expect{U(t+1) - U(t)| \.\pi(t) } \geq K > 0,
\end{equation}
then we have
\begin{equation}
     \mathbb{P}( U(t) < 0 ) < k e^{- t \frac{K^2 }{ 2 (B+K)^2 }  }  \quad \forall t > \frac{-U(0)}{K}
    \label{all_less}
\end{equation}
where $k =  e^{ -\frac{K U(0) }{ (B+K)^2 } }$.
\end{lemma}
\begin{proof}
By the positive drift, $U(t) - tK$ is also a submartingale with bounded difference  
\begin{equation}
    |U(t+1) - (t+1)K - (U(t) - tK)  | \leq  B+K,
\end{equation} for all $t\geq 0$. Applying Azuma's inequality \cite{Chung2006} on $U(t) - tK$, we have
\begin{equation}
\begin{aligned}
    &\mathbb{P}( U(t) < 0 ) \\
    &= \mathbb{P} \big( U(t) -  tK - U(0) < -   U(0) - t K   \big) \\
    & \leq \exp\[( -  \frac{ ( U(0) + tK )^2}{ 2t (B+K)^2  }  \]) \\
    & = \exp\[(- \frac{K^2 t}{ 2 (B+K)^2 }   \]) \exp \[( -\frac{K U(0) }{ (B+K)^2 } \])\\
    & \qquad \times \exp \[( - \frac{( U(0) )^2}{ 2t (B+K)^2 } \]) \\
    & \leq   e^{ -\frac{K U(0) }{ (B+K)^2 } - \frac{K^2 }{ 2 (B+K)^2 } t  } 
\end{aligned}
\end{equation}
for $t> \frac{-U(0)}{K}$, concluding the results.
\end{proof}

\begin{lemma} \label{lemmaS:EJSsortPM}
Using $sortPM$ with resolution $\delta$, the coarse binned sorted  log-likelihood $U_\alpha(t)$ defined by (\ref{eqS:alpha_pos}) and (\ref{eqS:alpha_U}) is a submartigale with respect to $\.\pi(t)$. In particular, we have 
\begin{equation}
\begin{aligned}
&\expect{ U_\alpha(t+1)  | \.\pi(t)} - U_\alpha(t) \geq  K_s  :=  \\
& \max \[\{  \frac{1}{2} D \[( \frac{1}{4} B_1 + \frac{3}{4} B_0 \Big\| B_0 \])  ,   \frac{1}{8} D\[( B_1 \Big \| \frac{3}{4} B_1 + \frac{1}{4} B_0 \]) \]\}
\end{aligned}
\end{equation}
for all $t>0$ where $B_1 = \text{Bern}(1-p[1\slash 2])$ and $B_0 = \text{Bern}(p[1\slash 2])$. 
\end{lemma}
\begin{proof}
Let $\sigma_t$ be the permutation such that $\sigma_t(\.\pi(t)) = \.\pi^{\downarrow}(t)$. To emphasize the effect of the different permutations at different time $t$, for a given permutation $\sigma$ we define $\.\pi^{\sigma}(t) := \sigma (\.\pi(t))$ and 
\begin{equation}
\begin{aligned}
     U_{\alpha}^\sigma(t) &:= U(\.\pi^{\alpha,\sigma}) \\
     &=\sum_{q=1}^{1\slash\alpha} \pi_{q}^{\alpha,\sigma}(t) \log \frac{\pi_{q}^{\alpha,\sigma}(t)}{1-\pi_{q}^{\alpha,\sigma}(t)},
\end{aligned}
\end{equation}
where
\begin{equation} 
   \pi_{q}^{\alpha,\sigma}(t) := \sum_{i\in \text{bin}(q)} \pi^{\sigma}_i(t), \ q = 1,2,..., 1\slash \alpha.
\end{equation}
By definition, we have $U_{\alpha}(t) \equiv U_{\alpha}^{\sigma_t}(t)$. Now, we can lower bound the expected drift as
\begin{equation} \label{eqS:sortPM_EJS}
    \begin{aligned}
        & \expect{ U_{\alpha}  (t+1) |  \.\pi(t)} - U_\alpha(t)  \\
        & = \expect{ U_{\alpha}^{\sigma_{t+1}} (t+1) |  \.\pi(t)} - U_\alpha^{\sigma_t}(t)\\
        & \overset{(a)}{\geq} \expect{ U_{\alpha}^{\sigma_t} (t+1) |  \.\pi(t) } - U_\alpha^{\sigma_t}(t) \\
        &\overset{(b)}{=}  \sum_{q=1}^{1\slash \alpha} \pi_{q}^{\sigma_t,\alpha}(t) D\[(   P^{\sigma_t}_{y_{t+1}|  \text{bin}(q) , S_{t+1} }   \Big\| P^{\sigma_t}_{y_{t+1}| \notin  \text{bin}(q), S_{t+1} } \]),
    \end{aligned}
\end{equation}
where
\begin{equation}
\begin{aligned}
       P^{\sigma_t}_{y_{t+1}|  \text{bin}(q) , S_{t+1}} &:=\frac{ 1 }{ \pi_{q}^{\alpha,\sigma_t}(t) } \sum_{i \in \text{bin}(q)} \pi_i^{\sigma_t}(t) P_{y_{t+1}|\sigma_t(i),S_{t+1}} \\
       P^{\sigma_t}_{y_{t+1}|  \notin \text{bin}(q), S_{t+1}} &:=\sum_{q' \neq q} \frac{\pi_{q'}^{\sigma_t,\alpha}(t)}{1-\pi_{q}^{\sigma_t,\alpha}(t)} P^{\sigma_t}_{y_{t+1}|  \text{bin}(q), S_{t+1}} 
\end{aligned}
\end{equation}
and $P_{y_{t+1}|\cdot,S_{t+1}}$ is as defined in (\ref{eqS:Pyt_original}). Here the inequality $(a)$ follows from $\.\pi^{\sigma_t,\alpha}(t+1) \succeq  \.\pi^{\sigma_{t+1},\alpha}(t~+~1~)$ and that $U(\.\pi)$ is Schur-convex with respect to $\.\pi$. And $(b)$ is a similar manipulation using Bayes's rule as was done in the proof of [Theorem 1 in~\cite{Naghshvar2015}]. 

We now further lower bound (\ref{eqS:sortPM_EJS}) by positivity and convexity of the KL divergence. We seperate the discussion into two cases:
\begin{enumerate}
    \item If $q^* = 1$:
    
    By the selection rule of $k^*$ in $sortPM$, we have $\pi_1^{\alpha}(t) \geq 1\slash 2$ and $\pi^{\sigma_t}_{[1,k^*]}(t) \geq 1\slash 4$. Therefore, 
    \begin{equation}
    \begin{aligned} \label{eqS:sortPM_case1}
       & (\ref{eqS:sortPM_EJS})  \geq  \\
       &\sum_{q=1}^{1\slash \alpha} \pi_{q}^{\sigma_t,\alpha}(t) D\[(   P^{\sigma_t}_{y_{t+1}|  \text{bin}(q) , S_{t+1} }   \Big\| P^{\sigma_t}_{y_{t+1}| \notin  \text{bin}(q), S_{t+1} } \]) \\
      & \overset{(c)}{\geq}  \pi_1^{\sigma_t,\alpha}(t)  D\[(   \frac{\pi^{\sigma_t}_{[1,k^*]}(t)}{\pi_1^{\alpha}(t)} B_1 + \frac{\pi^{\sigma_t}_{[k^*+1,\frac{\alpha}{\delta}]}(t)}{\pi_1^{\alpha}(t)} B_0   \ \Bigg\| \ B_0 \]) \\
      & \overset{(d)}{\geq}  \frac{1}{2} D \[( \frac{1}{4} B_1 + \frac{3}{4} B_0 \Big\| B_0 \]),
    \end{aligned}
    \end{equation}
    where (c) is by positivity of KL divergence and (d) is by  $\pi_1^{\alpha}(t) \geq 1\slash 2$ and $\pi^{\sigma_t}_{[1,k^*]}(t) \geq 1\slash 4$.

    \item If $q^*> 1$:
    
    By the selection rule of $k^*$ in $sortPM$, we have $\pi^{\sigma_t}_{[1,k^*]}(t) \leq 3\slash 4$. 
    
    WLOG, we assume that $k^*< \max \text{bin}(q^*)$ otherwise it reduces to the case of Lemma \ref{factS:EJSori_sort_dya}. Together with the selection rule of $k^*$, we have $\pi_{[1,q^*]}^{\sigma_t,\alpha} (t) \geq \frac{1}{2}$. By sorting we aslo have $\pi_{[1,q^*-1]}^{\sigma_t,\alpha} (t) \geq \pi_{q^*}^{\sigma_t,\alpha} (t)$. Therefore $\pi_{[1,q^*-1]}^{\sigma_t,\alpha} (t) \geq \frac{1}{4}$.  
    
    Now can proceed the lower bound as 
    \begin{equation}
    \begin{aligned} \label{eqS:sortPM_case2}
       & (\ref{eqS:sortPM_EJS}) \\
       &\geq \sum_{q=1}^{1\slash \alpha} \pi_{q}^{\sigma_t,\alpha}(t) D\[(   P^{\sigma_t}_{y_{t+1}|  \text{bin}(q) , S_{t+1} }   \Big\| P^{\sigma_t}_{y_{t+1}| \notin  \text{bin}(q), S_{t+1} } \]) \\
      & \overset{(e)}{\geq}   \pi_{[1,q^*-1]}^{\sigma_t,\alpha} (t) D( B_1 \|   \pi_{[1,k^*]}^{\sigma_t} B_1  + \pi_{[k^*+1,\frac{1}{\delta}]}^{\sigma_t} B_0) \\
      & \overset{(f)}{\geq}  \frac{1}{4} D\[( B_1 \Big \| \frac{3}{4} B_1 + \frac{1}{4} B_0 \]),
    \end{aligned}
    \end{equation}
    where $(e)$ is by Fact \ref{factS:JS_EJS} and positivity of the KL divergence, and (f) is from $\pi_{[1,q^*-1]}^{\sigma_t,\alpha} (t) \geq \frac{1}{4}$ and $\pi^{\sigma_t}_{[1,k^*]}(t) \leq 3\slash 4$. 
\end{enumerate}
Now let 
\begin{equation}
\begin{aligned}
    & K_s:= \\
    & \max \[\{  \frac{1}{2} D \[( \frac{1}{4} B_1 + \frac{3}{4} B_0 \Big\| B_0 \])  ,   \frac{1}{4} D\[( B_1 \Big \| \frac{3}{4} B_1 + \frac{1}{4} B_0 \]) \]\},
\end{aligned}
\end{equation}
and by (\ref{eqS:sortPM_case1}) and (\ref{eqS:sortPM_case1}), we conclude the assertion of this lemma.
\end{proof}

\begin{lemma} \label{lemmaS:EJShiePM}
Using $hiePM$ with resolution $\frac{1}{\delta}$, the nested log-likelihood $U^{\{l\}}(t)$ of lower resolution level $l<\log_2(1\slash \delta))$ defined in (\ref{eqS:nestedU}) is a submartigale. In particular, we have 
\begin{equation}
\begin{aligned}
    & \expect{ U^{\{l\}}(t+1) | \.\pi(t) } -  U^{\{l\}}(t)\\
    &\geq K_h := \min \Bigg\{  I \Big(\frac{1}{3},p[\frac{1}{2}] \Big), \frac{2}{3} D\Big(\frac{1}{3} B_1 + \frac{2}{3} B_0 \Big\| B_0 \Big) \Bigg\} \\
\end{aligned}
\end{equation}
for all $t>0$, for any $l<S$, where $B_1 = \text{Bern}(1-p[1\slash 2])$ and $B_0 = \text{Bern}(p[1\slash 2])$. 
\end{lemma}
\begin{proof}
Given any $l<S$, if the selected codeword $D^{(l_{t+1})}$ is such that $l_{t+1}\leq l$, by Fact \ref{factS:EJSori_hiePM} we conclude the results. If otherwise $l_{t+1}>l$, then we have $D^{(l_{t+1})} \subseteq \text{bin}(q_t) $ for some $q_t$.  For notational simplicity, let 
    $\rho \equiv \pi_{D^{(l_{t+1})}}(t) := \sum_{i\in D^{(l_{t+1})}} \pi_i(t)$
and $B^0 \equiv \text{Bern}(p[2^{-l_{t+1}}])$, $B^1 \equiv \text{Bern}(1-p[2^{-l_{t+1}}])$. We have
\begin{equation} 
\begin{aligned}
     &\expect{ U^{\{l\}}(t+1) | \.\pi(t) } -  U^{\{l\}}(t) \\
     &= \sum_{q=1}^{2^l} \pi^{\{l\}}_q(t) D\[( P_{y_{t+1}|  q , \gamma} \Big\| P_{y_{t+1}| \neq q, \gamma } \]) \\
     &\stackrel{(a)}{\geq} \frac{2}{3} D( \rho B^1 + (1-\rho) B^0 \| B^0 ) \stackrel{(b)}{\geq} \frac{2}{3} D( \frac{1}{3} B^1 + \frac{2}{3} B^0 \| B^0 ) \\
     &\geq \frac{2}{3} D\Big(\frac{1}{3} \text{Bern}(1-p[\frac{1}{2}]) + \frac{2}{3} \text{Bern}(p[\frac{1}{2}])  \Big\| \text{Bern}(p[\frac{1}{2}])  \Big) \Big\}.
\end{aligned}
\end{equation}
where (a) and (b) are by the selection rule of $hiePM$ that $\.\pi^{\{l\}}_{q_t}(t)>2\slash 3$ whenever $l_t > l$ and that $1 \slash 3 \leq \rho \leq 2 \slash 3$. This concludes the assertion.
\end{proof}

\begin{lemma} \label{lemmaS:EJSdyaPM}
Using $dyaPM$ with resolution $\delta$, the nested log-likelihood $U^{\{l\}}(t)$ of lower resolution level $l<\log_2(1\slash \delta))$ defined in (\ref{eqS:nestedU}) is a submartigale. In particular, we have 
\begin{equation} \label{eqS:dyaPM_Kd}
\begin{aligned}
    & \expect{ U^{\{l\}}(t+1) | \.\pi(t) } -  U^{\{l\}}(t)  \geq K_d :=  \\
   & \min \Bigg\{ \min_{\rho\in [0,1\slash 4]} \max \{ f(\rho),g(\rho)\} ,  \min_{\rho\in [1\slash 4,1\slash 2]} f(\rho), \\
   &\quad \qquad \frac{1}{4} D\[( \frac{1}{4}  B_1 + \frac{3}{4} B_0 \Big\|  B_0 \]) \Bigg\}  > 0, 
\end{aligned}
\end{equation}
for any $t>0$ and $l<\log_2(1\slash\delta)$, where 
$B_1 = \text{Bern}(1-p[1\slash 2])$, $B_0 = \text{Bern}(p[1\slash 2])$, 
\begin{equation}
f(\rho) =  \rho D\[( B_1 \big\|   (3 \slash 4) B_1 + (1 \slash 4) B_0  \])    
\end{equation}
\begin{equation}
    \begin{aligned}
        g(\rho) = (1\slash 2 - \rho)  
         & D\Big( (1-4\rho)B_1 + 4\rho B_0 \ \Big\| \\ &  (1\slash 2 + \rho)B_1 +  (1\slash 2 - \rho)B_0 \Big)
    \end{aligned}.
\end{equation}
\end{lemma}
\begin{proof}
By similar algebraic effort as in  [Theorem 1 in~\cite{Naghshvar2015}], the expected drift can be written as
\begin{equation} \label{eqS:nestedEJS}
\begin{aligned}
    &\expect{ U^{\{l\}}(t+1) | \.\pi(t) } -  U^{\{l\}}(t) \\
    & =  \sum_{q=1}^{2^l} \pi^{\{l\}}_q(t) D\[( P_{y_{t+1}|  \in \text{bin}(q) , \gamma} \Big\| P_{y_{t+1}| \notin \text{bin}(q), \gamma } \]),
\end{aligned} 
\end{equation}
where
\begin{equation}
\begin{aligned}
    &P_{y_{t+1} | \in \text{bin}(q) , \gamma } :=  \frac{1}{ \pi_q^{\{l\}}(t) } \\
    & \quad  \times \sum_{i\in \text{bin}(q)} \pi_i(t) p\big(y_{t+1} \big|\theta = i , S_{t+1} = \gamma(\.\pi(t)) \big) 
\end{aligned}
\end{equation} and 
\begin{equation}
\begin{aligned}
    &P_{y_{t+1} | \notin \text{bin}(q) , \gamma } :=  \frac{1}{\sum_{i\notin \text{bin}(q)} \pi_i(t) } \\
    & \quad  \times \sum_{i\notin \text{bin}(q)} \pi_i(t) p\big(y_{t+1} \big|\theta = i , S_{t+1} = \gamma(\.\pi(t)) \big) .
\end{aligned}
\end{equation}

We drop $(t)$ and write $\.\pi \equiv \.\pi(t)$ in the proof frequently for notational simplification. We write the starting index of $H^{m^*_t}_{l^*_t}$ as $d\equiv m^*_t 2^{L-l^*_t}$. Furthermroe, let the bin of level $l$ that contains $k^*$ be $q^*$, $i.e.$  $k^*\in \text{bin}(q^*)$ and $b_m = \min (\text{bin}(q^*))$ and $b_M = \max (\text{bin}(q^*))$. 

The case of $l=\log_2(1\slash\delta)$ is done by Lemma \ref{factS:EJSori_sort_dya}. For any given $l<\log_2(1\slash\delta)$, we separate into two cases: 
\begin{enumerate}
    \item $S_{t+1}=\gamma_d( \.\pi(t))$ contains at least one bin of level $l$, $i.e.$ bin$(q)\subseteq S_{t+1}$ for some $q$: 
    \begin{equation} \label{eqS:EJSdya_case1}
        \begin{aligned}
    &\expect{ U^{\{l\}}(t+1) | \.\pi(t) } -  U^{\{l\}}(t) \\
     &= \sum_{q=1}^{2^l} \pi^{\{l\}}_q(t) D\[( P_{y_{t+1}|  \in \text{bin}(q) , \gamma} \Big\| P_{y_{t+1}| \notin \text{bin}(q), \gamma } \]) \\
     & \stackrel{(a)}{\geq}  \max\Bigg\{ \pi_{[d,b_m-1]} D\[( B_1 \big\|   \pi_{[d,k^*]} B_1 + (1-\pi_{[d,k^*]}) B_0  \]) , \\
    &  \qquad \pi^{\{l\}}_{q^*}     D\Big(\frac{\pi_{[b_m,k^*]}}{\pi^{\{l\}}_{q^*}} B_1 + \frac{\pi_{[k^*+1,b_M]}}{\pi^{\{l\}}_{q^*}} B_0 \ \Big\|  \\ 
         & \qquad \qquad \qquad \pi_{[d,k^*]} B_1 + (1-\pi_{[d,k^*]}) B_0 \ \Big)  \Bigg\},
        \end{aligned}
    \end{equation}
    where we used 
    \begin{equation}
        \begin{aligned}
            &D\[( P_{y_{t+1}|  q , \gamma} \Big\| P_{y_{t+1}| \neq q, \gamma } \]) \geq D\[( P_{y_{t+1}|  q , \gamma} \Big\| P_{y_{t+1}|  \gamma } \]) \\
            &D\[( P_{y_{t+1}|  q , \gamma} \Big\| P_{y_{t+1}| \neq q, \gamma } \]) \geq 0
        \end{aligned}
    \end{equation}
    in (a). Note that by the binary tree construction of $H^{m}_{l}$, we have $ [b_m,k^*] \subseteq \text{bin}(q^*) \subseteq H^{2m^*_t}_{l^*_t+1}$. Therefore, 
    \begin{equation} \label{eqS:dya_proof_eq1}
    \pi_{[b_m,k^*]} \leq  \pi^{\{l\}}_{q^*} \leq    \pi_{H^{2m^*_t}_{l^*t+1}} \leq \frac{1}{2}.
    \end{equation}
    By the selection rule of $k^*$ and that $\pi_k\leq 1\slash 2$, we also know that $\pi_{[d,k^*]} \leq 3 \slash 4$. Together with (\ref{eqS:dya_proof_eq1}) we can lower bound the first part in (\ref{eqS:EJSdya_case1}) as
    \begin{equation}
    \begin{aligned}
        &\pi_{[d,b_m-1]} D\[( B_1 \big\|   \pi_{[d,k^*]} B_1 + (1-\pi_{[d,k^*]}) B_0  \]) \\
        & \geq \rho D\[( B_1 \big\|   (3 \slash 4) B_1 + (1 \slash 4) B_0  \]) := f(\rho) \\
    \end{aligned}
    \end{equation}
    where we used $\rho \equiv \pi_{[d,b_m-1]}$ for further simplification of the notation.
    
    On the other hand, without loss of generality we assume that $k^* < b_M$ (otherwise if $k^* = b_M$, it reduces to the case of Lemma \ref{factS:EJSori_sort_dya}). By the selection rule of $k^*$ and that $k^* < b_M$, we have 
    \begin{equation} 
        0 \leq \frac{1}{2} - \pi_{[d,k^*]} \leq \pi_{[d,b_M]} -  \frac{1}{2}
    \end{equation}
    which can be re-written as
    \begin{equation} \label{eqS:dya_proof_eq2}
        0 \leq \frac{1}{2} - \rho - \pi_{[b_m,k^*]} \leq \rho + \pi^{\{l\}}_{q^*}  -  \frac{1}{2}.
    \end{equation}
    Therefore,
    \begin{equation} \label{eqS:dya_proof_second_eq1}
        \begin{aligned}
            \frac{\pi_{[b_m,k^*]}}{\pi^{\{l\}}_{q^*}} & \stackrel{(b)}{\geq} \frac{1- \pi^{\{l\}}_{q^*} - 2\rho }{\pi^{\{l\}}_{q^*}} \\
            & = \frac{1- 2\rho }{\pi^{\{l\}}_{q^*}} -1  \stackrel{(c)}{\geq} 1-4\rho,
        \end{aligned}
    \end{equation}
    where (b) is by (\ref{eqS:dya_proof_eq2}) and (c) by (\ref{eqS:dya_proof_eq1}). And again by (\ref{eqS:dya_proof_eq1}) we also have
    \begin{equation} \label{eqS:dya_proof_second_eq2}
        \pi_{[d,k^*]} \leq \pi_{[d,b_M]} = \rho + \pi^{\{l\}}_{q^*} \leq \rho + \frac{1}{2}.
    \end{equation}
    With (\ref{eqS:dya_proof_eq2}), (\ref{eqS:dya_proof_second_eq1}) and (\ref{eqS:dya_proof_second_eq2}), the second part in equation (\ref{eqS:EJSdya_case1}) can then be lower bounded as
    \begin{equation}
    \begin{aligned}
        &\pi^{\{l\}}_{q^*}  D\Big( \frac{\pi_{[b_m,k^*]}}{\pi^{\{l\}}_{q^*}} B_1 + \frac{\pi_{[k^*+1,b_M]}}{\pi^{\{l\}}_{q^*}} B_0\  \Big\| \\
        & \qquad \qquad  \qquad \pi_{[d,k^*]} B_1 + (1-\pi_{[d,k^*]}) B_0 \Big) \\
        & \geq  (1\slash 2 - \rho)   D\Big(  (1-4\rho)B_1 + 4\rho B_0\ \Big\| \\
        & \quad \qquad  \qquad \qquad   (1\slash 2 + \rho)B_1 +  (1\slash 2 - \rho) B_0 \Big) \\
        &:= g(\rho).
    \end{aligned}
    \end{equation}
    Therefore (\ref{eqS:EJSdya_case1}) is lower-bounded by $K_d$ defined in (\ref{eqS:dyaPM_Kd}). It remains to show that $K_d>0$. Now, since $f(\rho) > 0$ is increasing for $\rho \in (0,1\slash 2]$ and $g(0) > 0$, we have
    \begin{equation}  \label{eqS:EJSdya_case1_final}
    \begin{aligned}
        &\min_{\rho\in (0,1\slash 4]} \max \{ f(\rho),g(\rho)\}  > 0\\
        &\min_{\rho\in (1\slash 4,1\slash 2]} f(\rho) > 0
    \end{aligned},
    \end{equation}
    concluding case 1.
    \item $S_{t+1}=\gamma_d( \.\pi(t))$ is within a bin of level $l$, $i.e.$  $S_{t+1} \subseteq \text{bin}(q^*)$: 
    \begin{equation} \label{eqS:EJSdya_case2}
      \begin{aligned}
     &\expect{ U^{\{l\}}(t+1) | \.\pi(t) } -  U^{\{l\}}(t) \\
     & \qquad= \sum_{q=1}^{2^l} \pi^{\{l\}}_q(t) D\[( P_{y_{t+1}|  q , \gamma} \Big\| P_{y_{t+1}| \neq q, \gamma } \]) \\
     & \qquad \geq  \pi^{\{l\}}_{q^*} D\[(\frac{\pi_{[b_m,k^*]}}{\pi^{\{l\}}_{q^*}} B_1 + \frac{\pi_{[k^*+1,b_M]}}{\pi^{\{l\}}_{q^*}} B_0 \Big\|  B_0 \])  \Big\}.
        \end{aligned}
    \end{equation}
    
    By the selection rule of $k^*$ and that $S_{t+1} \subseteq \text{bin}(q^*)$, we know that $\pi^{\{l\}}_{q^*} \geq \pi_{S_{t+1}} \geq 1\slash 4$ and that $\frac{\pi_{[b_m,k^*]}}{\pi^{\{l\}}_{q^*}} \geq \pi_{[b_m,k^*]} = \pi_{S_{t+1}} \geq 1\slash 4 $. Therefore,
    \begin{equation} \label{eqS:EJSdya_case2_final}
        (\ref{eqS:EJSdya_case2}) \geq \frac{1}{4} D\[( \frac{1}{4}  B_1 + \frac{3}{4} B_0 \Big\|  B_0 \]).
    \end{equation}
\end{enumerate}
The result is concluded by combining the two cases from (\ref{eqS:EJSdya_case1_final}) and (\ref{eqS:EJSdya_case2_final}).
\end{proof}

\end{document}